\documentclass[a4paper,10pt]{article}
\usepackage[margin=0pt]{geometry} 

\usepackage{longtable} 

\usepackage[T1]{fontenc}

\usepackage[]{amsmath, amssymb, amsthm, tabularx}

\usepackage[]{a4, xcolor, here}

\usepackage[]{pstricks, pst-text, pst-node, pst-tree, gastex}

\usepackage{tikz}

\usepackage[tikz]{bclogo}
\usepackage{comment}
\usepackage{multirow}
\usepackage{color, colortbl}

\usepackage{pgfplots}
\pgfplotsset{compat=1.15}
\usepackage{mathrsfs}
\usetikzlibrary{arrows}

\usepackage{boites,graphicx}
\usepackage{subcaption}
\usepackage{soul}

\definecolor{pastelyellow}{rgb}{0.99, 0.99, 0.59}
\definecolor{aqua}{rgb}{0.0, 1.0, 1.0} 
\definecolor{aquamarine}{rgb}{0.5, 1.0, 0.83} 
\definecolor{bananayellow}{rgb}{1.0, 0.88, 0.21}
\definecolor{burgundy}{rgb}{0.5, 0.0, 0.13}
\definecolor{ao(english)}{rgb}{0.0, 0.5, 0.0}
\definecolor{Gray}{gray}{0.9}

\setul{}{0.2ex}
\setulcolor{bananayellow}

\usepackage[normalem]{ulem}

\usepackage{mathrsfs}

\usepackage{stmaryrd}

\usepackage{enumerate}

\usepackage{paralist}

\usepackage{multienum}

\usepackage{multicol}

\usepackage{fancyhdr}

\usepackage{datetime}

\usepackage{everypage}
\usepackage[contents={},opacity=1,scale=1.6,
color=gray!90]{background}
\usepackage{ifthen}

\usepackage[]{hyperref} 

\hypersetup{colorlinks=true, linkcolor=red, anchorcolor = black, citecolor=blue, urlcolor=blue}

\newtheorem{theorem}{Theorem}[section]
\newtheorem{proposition}[theorem]{Proposition}
\newtheorem{lemma}[theorem]{Lemma}
\newtheorem{corollary}[theorem]{Corollary}

\theoremstyle{definition}
\newtheorem{definition}[theorem]{Definition}

\newtheorem{example}[theorem]{Example}
\newtheorem{remark}[theorem]{Remark}

\newcommand{\ba}{\mathbf{a}}
\newcommand{\bb}{\mathbf{b}}
\newcommand{\bc}{\mathbf{c}}

\newcommand{\bd}{\mathbf{d}}
\newcommand{\bD}{\mathbf{D}}

\newcommand{\be}{\mathbf{e}}

\newcommand{\bv}{\mathbf{v}}

\newcommand{\bo}{\mathbf{0}}

\newcommand{\cC}{\mathcal{C}}
\newcommand{\cD}{\mathcal{D}}

\newcommand{\cF}{\mathcal{F}}
\newcommand{\cG}{\mathcal{G}}

\newcommand{\cU}{\mathcal{U}}
\newcommand{\cV}{\mathcal{V}}
\newcommand{\cW}{\mathcal{W}}

\newcommand{\type}{\mathbf{t}}

\newcommand{\bbZ}{{\mathbb Z}}

\newcommand{\bbF}{{\mathbb F}} 

\renewcommand{\geq}{\geqslant}
\renewcommand{\leq}{\leqslant}

{\end{list}}%

\begin{document}

\renewcommand{\headrulewidth}{0pt}

\rhead{ }
\chead{\scriptsize {Flag Codes: Distance Vectors and Cardinality Bounds}}
\lhead{ }

\title{Flag Codes: Distance Vectors and \\ Cardinality Bounds
}

\author{\renewcommand\thefootnote{\arabic{footnote}}
Clementa Alonso-Gonz\'alez\footnotemark[1],\,  Miguel \'Angel Navarro-P\'erez\footnotemark[1], \\ 
\renewcommand\thefootnote{\arabic{footnote}} 
 Xaro Soler-Escriv\`a\footnotemark[1]}

\footnotetext[1]{Dpt.\ de Matem\`atiques, Universitat d'Alacant, 
Sant Vicent del Raspeig, Ap.\ Correus 99, E -- 03080 Alacant.
\\E-mail adresses: \href{mailto:clementa.alonso@ua.es}{clementa.alonso@ua.es}, \href{mailto:miguelangel.np@ua.es}{miguelangel.np@ua.es}, \href{mailto:xaro.soler@ua.es}{xaro.soler@ua.es}.}

{\small \date{\today}} 

\maketitle

\begin{abstract}
Given $\bbF_q$ the finite field with $q$ elements and an integer $n\geq 2$, a \emph{flag} is a sequence of nested subspaces of $\bbF_q^n$ and  a \emph{flag code} is a nonempty set of flags. In this context, the distance between flags is the sum of the corresponding subspace distances. Hence, a given flag distance value might be obtained by many different combinations. To capture such a variability, in the paper at hand, we introduce the notion of \emph{distance vector} as an algebraic object intrinsically associated to a flag code that encloses much more information than the distance parameter itself. Our study of the flag distance by using this new tool allows us to provide a fine description of the structure of flag codes as well as to derive bounds for their maximum possible size once the minimum distance and dimensions are fixed.
\end{abstract}



\section{Introduction}\label{sec: Introduction}

\emph{Network coding} was introduced in \cite{AhlsCai00} as a new method for sending information within networks modelled as acyclic multigraphs with possibly several senders and receivers, where intermediate nodes are allowed to send linear combinations of the received vectors, instead of simply routing them. In \cite{KoetKschi08}, the reader can find the first algebraic approach to network coding through non-coherent networks, i.e., those which their topology does not need to be known. In the same paper, K\"otter and Kschischang present \emph{subspace codes} as the most appropriate codes to this situation. To be precise, if $\bbF_q$ is the finite field of $q$ elements (with $q$ a prime power) and we consider a positive integer $n\geq 2$, a \emph{subspace code} is a nonempty collection of $\bbF_q$-vector subspaces of $\bbF_q^n$. When every codeword has the same dimension, say $1\leq k< n$, we speak about \emph{constant dimension codes}. In this context, we use the \emph{subspace distance}, denoted by $d_S$ (see \cite{KoetKschi08}). Constant dimension codes have been widely studied in the last decade. See, for instance, \cite{TrautRosen18} and the references therein.

Size and minimum distance are the most important parameters associated to an error-correcting code. The first one gives us the number of different messages that can be encoded. The second one is related with the error-correction capability of the code. According to this, there are two central problems when working with constant dimension codes. On the one hand, the study and construction of codes having the maximum possible distance for their dimension (see, for instance, \cite{GoManRo2012,GoRa2014,HorTraut2016,ManGorRos2008}). On the other hand, determining (or giving bounds for) the value $A_q(n, d, k)$, that is, the maximum possible size of a constant dimension code in $\cG_q(k, n)$ with minimum distance equal to $d$, is an interesting question that has led to many research works (see \cite{TablesSubspaceCodes,KohnKurz08,Kurz17,XiaFu09}, for instance).

In \cite{LiebNebeVaz2018}, the authors propose the use of \emph{flag codes} in network coding for the first time. This class of codes generalizes constant dimension codes and represents a possible alternative to obtain codes with good parameters in case that neither $n$ nor $q$ could be increased.  Flags are objects coming from classic linear algebra defined as follows. Given integers $1\leq t_1 < \dots < t_r < n$, a \emph{flag} of type $\type=(t_1, \dots, t_r)$ on $\bbF_q^n$ is a  sequence $\cF=(\cF_1,\dots, \cF_r)$ of nested subspaces $\cF_i$ of $\bbF_q^n$ such that $\dim(\cF_i)=t_i$, for every $1\leq i\leq r$. In this setting, codewords are flags of the prescribed type vector $\type$. As for constant dimension codes, describing the family of flag codes attaining the maximum possible distance for their type (\emph{optimum distance flag codes}) is a central question that has been addressed in \cite{CasoPlanarOrb, CasoPlanar, CasoGeneral, CasoGeneralOrb}. On the other hand, obtaining bounds for the value $A^f_q(n, d, \type)$, i.e., the maximum possible size of flag codes of type $\type$ on $\bbF_q^n$ and minimum distance equal to $d$, is also an important problem that, up to now, has only been addressed in the work \cite{Kurz20}, where the author focuses on the \emph{full type vector} $(1, \dots, n-1)$. The current paper represents a contribution in this direction. 

Given two flags $\cF, \cF'$ of type $\type$ on $\bbF_q^n$, their \emph{flag distance} is defined as $d_f(\cF, \cF')=\sum_{i=1}^r d_S(\cF_i, \cF'_i)$. This definition implies that a flag distance value might be obtained as different combinations of subspace distances and it suggests that the way we obtain the flag distance is relevant information to take into account beyond the proper numerical value. We deal with this question by introducing  the concept of \emph{distance vector} $\bd(\cF, \cF)=(d_S(\cF_1, \cF'_1), \dots, d_S(\cF_r, \cF'_r))$ associated to the pair of flags $\cF$ and $\cF'.$ This is an algebraic object strongly related not only to the value of the distance $d_f(\cF, \cF')$ but also to the nested linear structure of these flags. With its help, we develop techniques that allow us to study both the cardinality and the minimum distance of flag codes. First, the study of distance vectors will allow us to determine how the flag distance fluctuates when we consider flags sharing a given number of subspaces. Hence, we investigate the structure of flag codes in which different flags do not share simultaneously their subspaces of a prescribed set of dimensions. This approach leads us both to derive some structural properties of the flag code and to obtain upper bounds for the value $A_q^f(n,d,\type)$, strongly based on the maximum number of subspaces that different flags in a flag code of type $\type$ and distance $d$ can share. 

The paper is organized as follows. In Section \ref{sec: prelim}, we recall some definitions and known facts related with constant dimension codes and flag codes. In Section \ref{sec:  flag distance}, the notion of distance vector and a characterization of them are presented.  Section \ref{sec: remarkable values of the distance} is devoted to study the flag distance between flags of the same type that share certain subspaces. In Section \ref{sec: disjointness}, we generalize the notion of \emph{disjointness} introduced in \cite{CasoPlanar} and use the results obtained in the previous section in order to deduce structural properties of a code by simply looking at its minimum distance. In Section \ref{sec: bounds}, we apply the concepts and results in Sections \ref{sec: remarkable values of the distance} and \ref{sec: disjointness} to extract bounds for the values $A_q^f(n, d, \type)$. Last, Section \ref{sec: example} is dedicated to developing a very complete example that illustrates in detail the techniques previously discussed.


\section{Preliminaries}\label{sec: prelim}
In this section we recall some known facts on subspace and flag codes that we need in this paper. We start fixing some notation. Let $q$ be a prime power and consider the finite field $\bbF_q$ with $q$ elements. For every positive integer $n\geq 2$, we write $\bbF_q^n$ to denote the $n$-dimensional vector space over the field $\bbF_q$. Given a positive integer $k\leq n$, the \emph{Grassmann variety}, or simply the Grassmannian, of dimension $k$ is the set $\cG_q(k, n)$ of $k$-dimensional vector subspaces of $\bbF_q^n$. It is well known (see \cite{KoetKschi08}) that
\begin{equation}\label{eq: card grassmannian}
|\cG_q(k, n)|= \begin{bmatrix} n \\ k \end{bmatrix}_q :=   \frac{(q^n-1)\dots (q^{n-k+1}-1)}{(q^k-1)\dots (q-1)}. 
\end{equation}
The Grassmannian can be seen as a metric space endowed with the \emph{subspace distance} defined as
\begin{equation}
d_S(\cU, \cV)= \dim(\cU+\cV)-\dim(\cU\cap\cV)= 2(k - \dim(\cU\cap\cV)).
\end{equation}
for all $\cU, \cV\in\cG_q(k, n)$. A \emph{constant dimension code} $\cC$ in $\cG_q(k, n)$ is a nonempty collection of $k$-dimensional vector subspaces of $\bbF_q^n$. These codes were introduced in \cite{KoetKschi08} and studied many papers (see \cite{TrautRosen18} and references therein for further information). The \emph{minimum distance} of $\cC$ is the value
$$
d_S(\cC)=\min\{ d_S(\cU, \cV) \ | \ \cU, \cV\in\cC, \ \cU\neq \cV \},
$$
whenever $|\cC|\geq 2$. If $|\cC|= 1$, we put $d_S(\cC)=0$. In any case, the subspace distance is an even integer such that
\begin{equation}\label{eq: bound subspace distance}
0\leq d_S(\cC)\leq \left\lbrace
\begin{array}{lll}
2k     & \text{if} & 2k\leq n,\\
2(n-k) & \text{if} & 2k\geq n.
\end{array}
\right.
\end{equation}
The study and construction of constant dimension codes attaining this upper bound for the distance has been addressed in several papers (see \cite{GoRa2014,ManGorRos2008}, for instance). Another important problem is the one of determining (or giving bounds for) the value $A_q(n, d, k)$, which denotes the maximum possible size for constant dimension codes in $\cG_q(k, n)$ having prescribed minimum distance $d$. The reader can find constructions of constant dimension codes as well as lower and upper bounds for $A_q(n, d, k)$ in \cite{EtzVar2011,TablesSubspaceCodes,HorTraut2016,KohnKurz08,Kurz17,SilEtz2011,SilTraut2013,TrautRosdn2010,TrautRosen18,XiaFu09}.
As a generalization of constant dimension codes, in \cite{LiebNebeVaz2018}, the authors introduced the use of flag codes in network Coding. Let us recall some basic definitions in this matter.

Given integers $1\leq t_1 < \dots < t_r < n$, a \emph{flag} of type $\type=(t_1, \dots, t_r)$ on $\bbF_q^n$ is a sequence $\cF=(\cF_1, \dots, \cF_r)$ of nested subspaces $\cF_i$ of $\bbF_q^n$ such that $\dim(\cF_i)=t_i$, for every $1\leq i\leq r$. The vector $(1, \dots, n-1)$ is called  the \emph{full type vector} and flags of this type are known as \emph{full flags}.

Throughout the rest of the paper, we will write $\type$ to denote an arbitrary but fixed type vector $\type=(t_1, \dots, t_r).$ The \emph{flag variety}  $\cF_q(\type, n)$ is the set of all the flags of type $\type$ on $\bbF_q^n$. This variety contains exactly
\begin{equation}\label{eq: card flags}
|\cF_q(\type, n)| = \begin{bmatrix} n \\ t_1 \end{bmatrix}_q \begin{bmatrix} n -t_1 \\ t_2-t_1 \end{bmatrix}_q \cdots \begin{bmatrix} n-t_{r-1} \\ n-t_r \end{bmatrix}_q 
\end{equation}
elements (see \cite{Kurz20}) and it can be equipped with the \emph{flag distance}, computed as
\begin{equation}\label{def: flag dist}
d_f(\cF, \cF')= \sum_{i=1}^r d_S(\cF_i, \cF'_i),
\end{equation}
for every pair of flags $\cF, \cF'\in\cF_q(\type, n).$  

A \emph{flag code} $\cC$ of type $\type$ on $\bbF_q^n$ is a nonempty subset of $\cF_q(\type, n)$. We can naturally associate to it a family of $r$ constant dimension codes by projection. For every $1\leq i\leq r$, consider the map
\begin{equation}\label{def: projection}
p_i : \cF_q(\type, n) \longrightarrow \cG_q(t_i, n)
\end{equation}
defined as $p_i((\cF_1, \dots, \cF_r))=\cF_i,$ for every $(\cF_1,\dots, \cF_r)\in \cF_q(\type, n)$. With this notation, the \emph{$i$-th projected code} $\cC_i$ of the flag code $\cC$ is the constant dimension code $\cC_i=p_i(\cC)\subseteq \cG_q(t_i, n),$ consisting of all the $i$-th subspaces of flags in $\cC$.

If $\cC\subseteq\cF_q(\type, n)$ is a flag code with $|\cC|\geq 2$,  its \emph{minimum distance} is defined as
$$
d_f(\cC)=\min\{d_f(\cF, \cF') \ | \ \cF, \cF'\in\cC, \ \cF\neq \cF'\}
$$
and, if $|\cC|=1$, we put $d_f(\cC)=0$. Notice that, by means of (\ref{eq: bound subspace distance}), we can easily deduce that $d_f(\cC)$ is an even integer such that
\begin{equation}\label{eq: max dist general}
0 \leq d_f(\cC) \leq  2\left(\sum_{t_i \leq \left\lfloor \frac{n}{2}\right\rfloor} t_i + \sum_{t_i > \left\lfloor \frac{n}{2}\right\rfloor} (n-t_i)\right).
\end{equation}
When working with full flag codes, the previous bound becomes
\begin{equation}\label{eq: max distance full}
0 \leq d_f(\cC) \leq \left\lbrace
\begin{array}{cccl}
\frac{n^2}{2}   & \text{if} & n &  \text{is even}, \\
 & & & \\ [-1em]
\frac{n^2-1}{2} & \text{if} & n &  \text{is odd}.
\end{array}
\right.
\end{equation}

In the flag codes setting, we write $A_q^f(n, d, \type)$ to denote the maximum attainable size for a flag code in $\cF_q(\type, n)$ with minimum distance equal to $d$. In case of working with full flags, we drop the type vector and simply write $A_q^f(n, d)$. This notation was recently introduced by Kurz in \cite{Kurz20}. In that work, the author provided techniques to upper and lower bound these values in the full type case. Moreover, an exhaustive list of exact values of $A_q^f(n, d)$ is also given for small values of $n$.

\section{Flag distance versus distance vectors}\label{sec:  flag distance}
As seen in Section \ref{sec: prelim}, the flag distance extends, in some sense, the subspace distance. However, since it is defined as a sum, a particular flag distance value might be attained by adding different combinations of subspace distances. This makes that the minimum distance of a flag code will have associated some of the possible combinations (maybe all of them). In order to clarify this fact, in this section, we introduce the concept of \emph{distance vector} to better represent how the distance between different flags is distributed among their subspaces. 

\begin{definition}\label{def: distance vector associated to a pair of flags}
Given two different flags $\cF, \cF'$ of type $\type$ on $\bbF_q^n$, their associated \emph{distance vector} is
$$
\bd(\cF, \cF')=(d_S(\cF_1, \cF_1'), \dots, d_S(\cF_r, \cF_r')) \in 2\bbZ^r.
$$
\end{definition}
Notice that the sum of the components of $\bd(\cF, \cF')$ is the flag distance $d_f(\cF, \cF')$ defined in (\ref{def: flag dist}). Given a positive integer $n\geq 2$ and a type vector $\type$, we denote by $D^{(\type,n)}$ the maximum possible value of the flag distance in $\cF_q(\type,n)$ that, as a consequence of (\ref{eq: max dist general}), is
\begin{equation}\label{eq: D^(type,n)}
D^{(\type,n)}=2\left(\sum_{t_i \leq \left\lfloor \frac{n}{2}\right\rfloor} t_i + \sum_{t_i > \left\lfloor \frac{n}{2}\right\rfloor} (n-t_i)\right).
\end{equation}
In particular, when working with the full type vector, we simply write
\begin{equation}\label{eq: D^n}
D^n=\left\lbrace
\begin{array}{cccl}
\frac{n^2}{2}   & \text{if} & n &  \text{is even}, \\
 & & & \\ [-1em]
\frac{n^2-1}{2} & \text{if} & n &  \text{is odd}
\end{array}
\right.
\end{equation} 
to denote the maximum possible distance between full flags on $\bbF_q^n$ (see (\ref{eq: max distance full})). For technical reasons, even if we work with $n \geq 2$, we extend this definition to the case $n=1$ and put $D^1=0$.

From now on, we write $d$ to denote an even integer such that $0\leq d\leq D^{(\type,n)}$. Observe that, under these conditions, we can always find flags $\cF, \cF'\in\cF_q(\type,n)$ such that $d_f(\cF, \cF')=d$. Hence, such a value $d$ represents the possible values for the flag distance in $\cF_q(\type,n)$. Let us study in which ways this distance value $d$ can be obtained.

\begin{definition}
Let $d$ be an even integer such that $0\leq d\leq D^{(\type,n)}$. We define \emph{the set of distance vectors associated to $d$ for the flag variety $\cF_q(\type, n)$} as 
$$
\cD(d, \type, n) = \{ \bd(\cF, \cF') \ | \ \cF, \cF'\in\cF_q(\type, n), \ d_f(\cF, \cF')= d  \} \subseteq 2\bbZ^r.
$$
On the other hand, the set of \emph{distance vectors for the flag variety $\cF_q(\type, n)$} is 
$$
\cD(\type, n) = \{ \bd(\cF, \cF') \ | \ \cF, \cF'\in\cF_q(\type, n)\} \subseteq 2\bbZ^r.
$$
and it holds
$$
\cD(\type, n) = \bigcup_{d} \cD(d, \type, n),
$$
where $d$ takes all the even integers between $0\leq d\leq D^{(\type, n)}$. When working with the full flag variety, we drop the type vector and simply write  $\cD(d, n)$ and $\cD(n),$ respectively.
\end{definition}

The next result reflects that, for every choice of the type vector, the set $\cD(\type, n)$ can be obtained from $\cD(n)$ by using the projection
\begin{equation}\label{eq: map pi type}
\begin{array}{rcccc}
\pi_{\type} & : &  \mathbb{Z}^{n-1} & \longrightarrow & \mathbb{Z}^r\\
			&   &  (v_1, \dots, v_{n-1}) &\longmapsto & (v_{t_1}, \dots, v_{t_r}).
\end{array}
\end{equation}

\begin{proposition}\label{prop: projection dist vect}
Consider a type vector $\type$ and the projection map $\pi_{\type}$ defined in (\ref{eq: map pi type}). It holds
$$
\pi_{\type}(\cD(n))= \cD(\type, n),
$$
i.e., distance vectors for an arbitrary flag variety can be obtained by projection from (possibly several) distance vectors for the full flag variety.  
\end{proposition}
\begin{proof}

Take a distance vector $\bd(\cF, \cF')\in\cD(n)$, for a pair of full flags $\cF=(\cF_1, \dots, \cF_{n-1})$ and $\cF'=(\cF'_1, \dots, \cF'_{n-1})$. It suffices to see that $\pi_{\type}(\bd(\cF, \cF'))$ is the distance vector associated to the pair of flags $\bar{\cF}=(\cF_{t_1}, \dots, \cF_{t_r})$ and $\bar{\cF}'=(\cF'_{t_1}, \dots, \cF'_{t_r})$ in $\cF_q(\type, n)$.
 
Conversely, given two flags $\bar{\cF}=(\bar{\cF}_1, \dots, \bar{\cF}_r)$ and $\bar{\cF}'=(\bar{\cF}'_1, \dots, \bar{\cF}'_r)$ in $\cF_q(\type, n)$, we can consider full flags $\cF=(\cF_1, \dots, \cF_{n-1})$ and $\cF'=(\cF'_1, \dots, \cF'_{n-1})$ such that $\cF_{t_i}=\bar{\cF}_i$ and $\cF'_{t_i}=\bar{\cF}'_i$, for all $1\leq i\leq r$. In this case, it holds $\pi_{\type}(\bd(\cF, \cF'))= \bd(\bar{\cF}, \bar{\cF}')$.
\end{proof}

\begin{remark}\label{rem: D type n and D n}
Notice that, for  every even integer $d$ such that $0 \leq d \leq D^{(\type, n})$, the set $\cD(d,\type, n)$ is nonempty.  Moreover, for some values of $d$, the set $\cD(d,\type, n)$ is reduced to just one element. For instance, if we take $d=0$, it holds $\cD(0, \type, n)=\{ \bo\}.$ If $d=D^{(\type, n)}$, there is also a unique distance vector, that we denote by $\bD^{(\type, n)}$. For every $1\leq i\leq r$, its $i$-th component $D^{(\type, n)}_i$ is exactly
\begin{equation}\label{eq: components vector D type n}
D^{(\type, n)}_i= \min\{2t_i, 2 (n-t_i)\},
\end{equation}
i.e., the maximum possible distance between $t_i$-dimensional subspaces of $\bbF_q^n$. Observe that, in particular, the distance vector $\bD^{(\type, n)}$ does not have any zero component. As before, when working with the full type vector, we simply write $\bD^n=(D^n_1, \dots, D^n_{n-1})$ to denote the unique distance vector associated to the maximum possible flag distance $D^n$, given in (\ref{eq: D^n}). Its components are $D^n_i=\min\{2i, 2(n-i)\}$, for $1\leq i\leq n-1$. In other cases, the set $\cD(d,\type, n)$ might contain more than one element, as we can see in Example \ref{ex: dv d=4}.
\end{remark}

Using the projection defined in (\ref{eq: map pi type}), and arguing as in Proposition \ref{prop: projection dist vect}, the next result follows straightforwardly.
\begin{corollary}
Consider a positive integer $n$ and fix a type vector $\type$ for $\bbF_q^n$. It holds
$$
\pi_{\type}(\bD^n)= \bD^{(\type, n)}.
$$
\end{corollary}

In the following definition we collect the subset of distance vectors of $\cD(\type,n)$ that are significant for a flag code in $\cF_q(\type, n)$.
\begin{definition}\label{def: distance vectors of a flag code}
Given a flag code $\cC\subseteq\cF_q(\type, n)$,  its \emph{set of distance vectors} is 
$$
\cD(\cC) = \{ \bd(\cF, \cF') \ | \ \cF, \cF'\in\cC, \ d_f(\cF, \cF')=d_f(\cC)\}.
$$
\end{definition}

\begin{remark}
In general, given a flag code $\cC\subseteq\cF_q(\type, n)$ and a pair of flags $\cF, \cF'\in \cC$ such that $d_f(\cC)=d_f(\cF, \cF'),$ it holds
$$
\bd(\cF, \cF') \in \cD(\cC) \subseteq \cD(d_f(\cC), \type, n)\subseteq \cD(\type, n).
$$
\end{remark}

\begin{example}\label{ex: dv d=4}
Let $\{\be_1, \be_2, \be_3, \be_4\}$ be the standard $\bbF_q$-basis of $\bbF_q^4$ and consider the following full flags on $\bbF_q^4$.
$$
\begin{array}{ccccc} 
\mathcal{F}^1 &=& (\left\langle \be_1 \right\rangle, & \left\langle \be_1, \be_2 \right\rangle , & \left\langle \be_1, \be_2, \be_4 \right\rangle),\\
\mathcal{F}^2 &=& ( \left\langle \be_2 \right\rangle, & \left\langle \be_1, \be_2 \right\rangle , & \left\langle \be_1, \be_2, \be_3 \right\rangle), \\
\mathcal{F}^3 &=& ( \left\langle \be_1 \right\rangle , &  \left\langle \be_1, \be_3 \right\rangle, & \left\langle \be_1, \be_2, \be_3 \right\rangle),\\
\mathcal{F}^4 &=& ( \left\langle \be_2 \right\rangle, & \left\langle \be_2, \be_3 \right\rangle , & \left\langle \be_1, \be_2, \be_3 \right\rangle).
\end{array}
$$
Notice that $D^4=16/2=8$. Thus, the possible values of the flag distance for full flags on $\bbF_q^4$ are all the even integers $d\in[0,8]$. In particular, for $d=4$, vectors $(2,0,2), (0,2,2), (2,2,0)$ are elements in $\cD(4, 4)$ since 
$$\bd(\cF^1, \cF^2)=(2,0,2),\ \bd(\cF^1, \cF^3)= (0,2,2)\  \text{and} \ \bd(\cF^2, \cF^3)=(2,2,0).$$
On the other hand, if we take the full flag code $\cC=\{\cF^1, \cF^2, \cF^4\}$, it holds
$$
\begin{array}{ccccccccccc}
 d_f(\cF^1, \cF^3) & = & 0 &+& 2 &+& 2 & = & 4, \\
 d_f(\cF^1, \cF^4) & = & 2 &+& 2 &+& 2 & = & 6,\\
 d_f(\cF^3, \cF^4) & = & 2 &+& 2 &+& 0 & = & 4.
\end{array}
$$
Hence, the distance of the code is $d_f(\cC)=4$ and $\cD(\cC)=\{ (0,2,2), (2,2,0)\}\subsetneq \cD(4, 4).$
\end{example}

Up to now, to show that a given vector $\bv\in 2\bbZ^r$ is a distance vector in $\cD(\type, n)$, we need to exhibit a pair of flags $\cF, \cF'\in\cF_q(\type ,n)$ such that $\bv=\bd(\cF, \cF')$. We finish the section with the next result that characterizes distance vectors in terms of some properties satisfied by their components.

\begin{theorem}\label{prop: properties distance vectors}
Let $d$ be an even integer such $0\leq d\leq D^{(\type, n)}$. A vector $\bv=(v_1, \dots, v_r)$ is a distance vector in $\cD(d,\type, n)$ if, and only if, the following statements hold:
\begin{enumerate}[(i)]
\item $\sum_{i=1}^r v_i = d,$  \label{prop: properties distance vectors-item1}
\item $v_i\in 2\bbZ,$ for all $1\leq i\leq r$,  \label{prop: properties distance vectors-item2}
\item $0\leq v_i \leq \min\{ 2t_i, 2(n-t_i)\},$ for every $1\leq i\leq r$, and \label{prop: properties distance vectors-item3}
\item $|v_{i+1} - v_i| \leq 2(t_{i+1}-t_i),$ for $1\leq i \leq r-1$. \label{prop: properties distance vectors-item4}
\end{enumerate}
\end{theorem}
\begin{proof}
We start assuming that $\bv\in\cD(d,\type, n)$. Statements  (\ref{prop: properties distance vectors-item1}), (\ref{prop: properties distance vectors-item2}) and (\ref{prop: properties distance vectors-item3}) follow from the definition of $\cD(d,\type, n)$. Let us prove (\ref{prop: properties distance vectors-item4}). Since $\bv\in\cD(d,\type, n)$, there must exist flags $\cF, \cF'\in \cF_q(\type, n)$ such that $d=d_f(\cF, \cF')$ and $\bv=\bd(\cF, \cF')$, i.e., $v_i=d_S(\cF_i, \cF'_i)= 2(t_i-\dim(\cF_i\cap\cF'_i)),$ for every $1\leq i\leq r$. Notice that, for every $1\leq i\leq r-1$, it holds

$$
2t_i-\dim(\cF_i\cap\cF'_i) = \dim(\cF_i+\cF'_i) \leq  \dim(\cF_{i+1}+\cF'_{i+1}) = 2t_{i+1} - \dim(\cF_{i+1}\cap\cF'_{i+1})
$$
and, as a consequence, 
\begin{equation}\label{eq: relation nested intersection}
\dim(\cF_i\cap\cF'_i) \leq  \dim(\cF_{i+1}\cap\cF'_{i+1}) \leq  \dim(\cF_i\cap\cF'_i) + 2(t_{i+1}-t_i).
\end{equation}
Moreover, we have that
\begin{equation}\label{eq: v_(i+1)-v_i}
v_{i+1}-v_i = 2(t_{i+1}-t_{i}) - 2 (\dim(\cF_{i+1}\cap\cF'_{i+1})- \dim(\cF_i\cap\cF'_i)).
\end{equation}
Hence, by using the first inequality of (\ref{eq: relation nested intersection}), we clearly obtain $v_{i+1}-v_i\leq 2(t_{i+1}-t_{i})$. On the other hand, combining the second inequality of (\ref{eq: relation nested intersection}) and (\ref{eq: v_(i+1)-v_i}), we get
$$
v_{i+1}-v_i\geq  2(t_{i+1}-t_{i}) - 4(t_{i+1}-t_{i}) = -2(t_{i+1}-t_{i})
$$
and (\ref{prop: properties distance vectors-item4}) holds.

Let us prove the converse. To do so, assume that $\bv=(v_1, \dots, v_r)$ is a vector satisfying conditions (\ref{prop: properties distance vectors-item1})-(\ref{prop: properties distance vectors-item4}). We want to show that $\bv\in\cD(d,\type, n)$ or, equivalently, to find a pair of flags in $\cF_q(\type,n)$ such that $\bv=\bd(\cF, \cF')$ and $d_f(\cF, \cF')=d$.

First, by means of (\ref{prop: properties distance vectors-item2}) and (\ref{prop: properties distance vectors-item3}), every $v_i$ is an even integer such that $0\leq v_i \leq \min\{2t_i, 2(n-t_i)\}$. Hence, each $v_i$ is an admissible distance value between $t_i$-dimensional subspaces of $\bbF_q^n$. Moreover, we can write every $v_i=2w_i$ for some integer $w_i$.

Notice that finding subspaces $\cF_i, \cF'_i\in\cG_q(t_i, n)$ with distance $d_S(\cF_i, \cF'_i)=v_i$ is equivalent to choose them  satisfying $\dim(\cF_i\cap\cF'_i)=t_i-w_i$. This can be clearly done for every $1\leq i\leq r$. However, we need that the chosen subspaces form flags $\cF=(\cF_1, \dots, \cF_r)$ and $\cF'=(\cF'_1, \dots, \cF'_r)$. We use an inductive process in order to construct such flags. We start taking subspaces  $\cF_1, \cF'_1\in\cG_q(t_1, n)$ such that $\dim(\cF_1 \cap \cF'_1)=t_1-w_1$. Assume now that, for some $1\leq i< r$, we have found subspaces $\cF_j, \cF'_j\in\cG_q(t_j, n)$, for all $1\leq j\leq i$, such that
$$
\begin{array}{ccccccc}
\cF_1  & \subsetneq & \dots & \subsetneq & \cF_{i-1}  & \subsetneq & \cF_i, \\
\cF'_1 & \subsetneq & \dots & \subsetneq & \cF'_{i-1} & \subsetneq & \cF'_i, 
\end{array}
$$
and $d_S(\cF_j, \cF'_j)= v_j =2w_j.$ Let us see that we can find suitable subspaces $\cF_{i+1}$ and $\cF_{i+1}'$. To do this, notice that, by using  (\ref{prop: properties distance vectors-item4}) and (\ref{prop: properties distance vectors-item3}), in this order, we obtain 
$$
\dim(\cF_i+\cF'_i)=t_i+w_i \leq t_{i+1}+w_{i+1} \leq n. 
$$
Thus, we can consider a subspace $\cU\in\cG_q(t_{i+1}+w_{i+1}, n)$ such that $\cF_i+\cF'_i\subseteq \cU$. It holds
$$
\dim(\cU)-\dim(\cF_i+\cF'_i)= t_{i+1}+w_{i+1}-(t_i+w_i) > w_{i+1}-w_i.
$$

We distinguish two possible situations in terms of the value $l_i:=w_{i+1}-w_i$. 

\noindent $\bullet$ If $l_i\geq 0$, then we put $m_i:=(t_{i+1} - t_i)-l_i$. Observe that, by means of (\ref{prop: properties distance vectors-item4}), we have that $m_i\geq 0$. Moreover, we have
$$
\begin{array}{ccl}
\dim(\cF_i+\cF'_i) + 2l_i + m_i  & = & (t_i + w_i) + 2l_i + m_i                        \\
 		                         & = & (t_i + w_i) + (w_{i+1} -w_i) + (t_{i+1} - t_i)  \\ 
		                         & = &  t_{i+1} + w_{i+1}                              \\
		                         & = & \dim(\cU).
\end{array}
$$
Hence, we can find linearly independent vectors $\ba_1, \dots, \ba_{l_i}, \bb_1, \dots, \bb_{l_i}, \bc_1, \dots, \bc_{m_i}$ in $\cU$ such that this subspace can be expressed as the direct sum
$$
\cU = (\cF_i+\cF'_i) \oplus \langle \ba_1, \dots, \ba_{l_i} \rangle \oplus \langle \bb_1, \dots, \bb_{l_i} \rangle \oplus \langle \bc_1, \dots, \bc_{m_i} \rangle.
$$ 
Now, consider the subspaces
$$
\begin{array}{ccccccc}
\cF_{i+1}  &:=& \cF_i  & \oplus & \langle \ba_1, \dots \ba_{l_i} \rangle  & \oplus & \langle \bc_1, \dots, \bc_{m_i}\rangle, \\
\cF'_{i+1} &:=& \cF'_i  & \oplus & \langle \bb_1, \dots, \bb_{l_i} \rangle & \oplus & \langle \bc_1, \dots, \bc_{m_i}\rangle,
\end{array}
$$
which have dimension
$$
\dim(\cF_{i+1})=\dim(\cF'_{i+1})= t_i + l_i + m_i = t_{i+1}.
$$
It is clear that $\cF_i\subsetneq \cF_{i+1}$ and $\cF'_i\subsetneq \cF'_{i+1}$. Moreover, observe that $\cF_{i+1}+\cF'_{i+1}=\cU$. Hence, $\dim(\cF_{i+1}+\cF'_{i+1})= \dim(\cU)= t_{i+1}+w_{i+1}$ and, consequently, it holds $\dim(\cF_{i+1}\cap\cF'_{i+1})= t_{i+1}-w_{i+1}$. As a result, we obtain $d_S(\cF_{i+1}, \cF'_{i+1})= v_{i+1}$, as desired.

\noindent $\bullet$ If $l_i < 0$, then it holds $t_i < t_i - l_i \leq t_i + w_i = \dim(\cF_i+\cF'_i)$. Thus, we can consider $(t_i - l_i)$-dimensional subspaces $\cV$ and $\cV'$ such that 
$$
\begin{array}{ccccccc}
\cF_i  &\subsetneq &\cV  &\subseteq &\cF_i+\cF'_i &\subseteq &\cU,\\
\cF'_i &\subsetneq &\cV' &\subseteq &\cF_i+\cF'_i &\subseteq &\cU.
\end{array}
$$
Notice that $\cV+\cV'=\cF_i+\cF'_i$. Besides, recall that
$$
\dim(\cU)-\dim(\cF_i+\cF'_i)= (t_{i+1}+w_{i+1})-(t_i+w_i)=(t_{i+1}-t_{i})+l_i\geq 0 
$$
since $\bv$ satisfies condition (\ref{prop: properties distance vectors-item4}). Hence, there exists a subspace $\cW\subseteq \cU$ of dimension $(t_{i+1}-t_{i})+l_i$ such that
$$
\cU= (\cF_i+\cF'_i) \oplus \cW = (\cV +\cV') \oplus \cW.
$$
Let us consider the subspaces
$$
\cF_{i+1} = \cV \oplus \cW \ \text{and} \ \cF'_{i+1} = \cV' \oplus \cW, 
$$
which have dimension 
$$
\dim(\cV)+\dim(\cW)= \dim(\cV')+\dim(\cW)=  (t_i - l_i)+(t_{i+1}-t_{i}+l_i)= t_{i+1}
$$
and clearly contain $\cF_i$ and $\cF'_i$, respectively. Moreover, since $\cF_{i+1}+\cF'_{i+1}=\cU$, we conclude that $\dim(\cF_{i+1}\cap \cF'_{i+1})= 2t_{i+1}-\dim(\cU)= 2t_{i+1}-(t_{i+1}+w_{i+1})= t_{i+1}-w_{i+1}$. This is equivalent to say that $d_S(\cF_{i+1}, \cF'_{i+1})= 2w_i=v_i$, as we wanted to prove.

In both cases, we conclude the existence of $\cF, \cF'\in\cF_q(\type, n)$ such that $\bv=\bd(\cF, \cF')$, which finishes the proof.
\end{proof}

\begin{example}\label{ex: not a dist vector} 
Consider the full flag variety on $\bbF_q^7$. In this case, $D^7=24$ and we can consider the possible value of the distance $d=20$. According to Theorem \ref{prop: properties distance vectors}, the set of distance vectors associated to $d=20$ is given by
$$
\cD(20, 7)=\{ (2,4,4,4,4,2), (2,4,6,4,2,2), (2,2,4,6,4,2) \}.
$$
Observe that, even though all the components of the vector $(2, \mathbf{2}, \mathbf{6}, 4, 4, 2)$ are allowed distances between subspaces of the corresponding dimensions and they sum $d=20$, such a vector is not a distance vector in $\cD(20, 7)$. This is due to the fact that the sequence $(\mathbf{2}, \mathbf{6})$ violates condition (\ref{prop: properties distance vectors-item4}) in Theorem \ref{prop: properties distance vectors}, since $6-2 = 4 > 2 = 2(3-2)$.

For $n=7$ and $\type=(1,3,5,6)$, we have $D^{(\type, 7)}=14$. Observe that the distance $d=12$ can only be attained by distance vectors
$$
\cD(12, \type, 7)= \{ (2,4,4,2), (2,6,2,2) \}.
$$
In this case consecutive components $(\mathbf{2}, \mathbf{6})$ in the vector$(\mathbf{2}, \mathbf{6},2,2)$ are allowed, since they represent distance between nested subspaces of dimensions $t_1=1$ and $t_2=3$. Hence, the difference $6-2=4= 2(t_2-t_1)$ respects the condition (\ref{prop: properties distance vectors-item4}) in Theorem \ref{prop: properties distance vectors}.
\end{example}

\section{Distance between flags sharing subspaces}\label{sec: remarkable values of the distance}

This section is devoted to the study of the flag distance between flags in $\cF_q(\type,n)$ that share subspaces. To do this, we start by analyzing the distance associated to distance vectors with a prescribed component, in particular, the ones having a component equal to zero. Then, we extend our study to distance vectors having several zeros among their components. This study will be used in Sections \ref{sec: disjointness} and \ref{sec: bounds} to obtain some information about the structure of flag codes as well as bounds for their cardinality depending on their minimum distance.

\subsection{Distance vectors with a fixed component}\label{subsec: fixed component}

We start by describing the interval of attainable distances by distance vectors in $\cD(\type,n)$ with their $i$-th component fixed, for some $1\leq i\leq r$. Throughout the rest of the section, we will write $v$ to denote an even integer $0\leq v\leq \min\{2t_i, 2(n-t_i)\}$. In other words, the integer $v$ represents a possible value for the distance between $t_i$-dimensional subspaces of $\bbF_q^n$. We focus on the set of distance vectors $\bv=(v_1,\dots,v_r)\in\cD(\type, n)$ with $v$ as their $i$-th component, paying special attention to those associated to the maximum and minimum distances.

Notice that, if we require a distance vector $\bv$ to satisfy $v_i=v$, then, by using condition (\ref{prop: properties distance vectors-item4}) in Theorem \ref{prop: properties distance vectors}, we obtain $|v_j-v|\leq 2|t_i-t_j|$ for all $1\leq j\leq r$. Moreover, by means of (\ref{prop: properties distance vectors-item2})-(\ref{prop: properties distance vectors-item4}) in Theorem \ref{prop: properties distance vectors}, for every $1\leq j\leq r$, the component $v_j$ must hold
\begin{equation}\label{eq: min and max dist provided a_i}
\max\{0, v - 2|t_i-t_j| \} \leq v_j \leq \min\{ 2t_j, 2(n-t_j), v+2|t_i-t_j|\}.
\end{equation}

\begin{definition}
Given $v$ as above, we write $d(i;v)^{(\type, n)}$ (resp. $D(i;v)^{(\type, n)}$) to denote the minimum (resp. maximum) distance that can be attained by distance vectors in $\cD(\type, n)$ having its $i$-th component equal to $v$.  According to (\ref{eq: min and max dist provided a_i}),  there exists a unique distance vector, that we denote by $\bd(i; v)^{(\type, n)}$ (resp. $\bD(i; v)^{(\type, n)}$), giving the distance $d(i; v)^{(\type, n)}$ (resp. $D(i;v)^{(\type, n)}$) and having $v$ as its $i$-th component. For every $1\leq j\leq r$, the $j$-th components of these vectors are given by
\begin{equation}\label{eq: components d(i, a) and D(i,a)}
\begin{array}{ccl}
d(i; v)_j^{(\type, n)} & = & \max \{ 0  , v - 2 |t_i - t_j | \},\\
D(i; v)_j^{(\type, n)} & = & \min \{ 2t_j, 2(n-t_j), v + 2 |t_i - t_j | \}.
\end{array}
\end{equation}
Consequently, the value $d(i;v)^{(\type, n)}$ (resp. $D(i; v)^{(\type, n)}$) is obtained as the sum of the components of $\bd(i; v)^{(\type, n)}$ (resp. $\bD(i; v)^{(\type, n)}$), given in (\ref{eq: components d(i, a) and D(i,a)}). Notice that, by construction, these values satisfy $0\leq d(i;v)^{(\type, n)} \leq D(i;v)^{(\type, n)} \leq D^{(\type; n)}$. When working with the full type variety, we simply write $d(i; v)^n$, $D(i; v)^n$, $\bd(i; v)^n$ and $\bD(i; v)^n$.
\end{definition}

\begin{example}
For the full flag variety on $\bbF_q^7$, take $i=3$ and $v=4$, we have 
$$
\bd(3; 4)^7= (0,2,\mathbf{4},2,0,0) \ \text{and} \ \ \bD(3; 4)^7= (2,4,\mathbf{4},6,4,2)
$$
and their associated distances are the values $d(3;4)^7=8$ and $D(3; 4)^7=22$.

Consider now the type vector $\type=(1,3,5,6)$ on $\bbF_q^7$. For the same choice of $i=3$ and $v=4$, we have
$$
\bd(3; 4)^{(\type, 7)}= (0,0,\mathbf{4},2) \ \text{and} \ \ \bD(3; 4)^{(\type, 7)}= (2,6,\mathbf{4},2).
$$
Hence, in this case, we have $d(3;4)^{(\type, 7)}=6$ and  $D(3; 4)^{(\type, 7)}=14 = D^{(\type, 7)}$.
\end{example}

According to the definition of $d(i;v)^{(\type, n)}$ and $D(i; v)^{(\type, n)}$, it is clear that if $\bv\in\cD(\type, n)$ such that $v_i=v$, then its associated distance $d$ is an even integer such that $d(i;v)^{(\type, n)}\leq d\leq D(i;v)^{(\type, n)}$. The next result allows us to ensure that the converse is also true.

\begin{proposition}\label{prop: every distance between dmin dmax}
Consider the type vector $\type$ on $\bbF_q^n,$ take an index $1\leq i\leq r$ and an even integer $0\leq v\leq \min\{2t_i, 2(n-t_i)\}$. If  $d$ is an even integer such that $d(i;v)^{(\type, n)}\leq d\leq D(i;v)^{(\type, n)}$, then there exist distance vectors in $\cD(d,\type, n)$ with $v$ as its $i$-th component.
\end{proposition}
\begin{proof}
We prove the result by induction on $d$. For $d=d(i;v)^{(\type, n)}$, the result holds since the vector $\bd(i;v)^{(\type, n)}$ satisfies the required condition. Now, assume that, for some even integer $d$ such that $d(i;v)^{(\type, n)}\leq d < D(i;v)^{(\type, n)}$, we have found a distance vector $\bv=(v_1, \dots, v_r)\in\cD(d,\type, n)$ such that $v_i=v$. Let us use $\bv$ to construct a suitable distance vector in $\cD(d+2, \type, n)$. Observe that, since $d< D(i;v)^{(\type, n)}$, clearly the set
$$
\{ v_j \ | \ v_j<  D(i;v)^{(\type, n)}_j \}
$$
is nonempty. Hence, we can consider the minimum $1\leq k\leq r$ such that $v_k =  \min\{ v_j \ | \ v_j<  D(i;v)^{(\type, n)}_j \}$. According to this, we have that $v_k<  D(i;v)^{(\type, n)}_k$ and then $v_k+2\leq  D(i;v)^{(\type, n)}_k$ is a possible distance between $t_k$-dimensional subspaces of $\bbF_q^n$. Consider now the $k$-th canonical vector $\be_k\in\bbZ^r$, i.e., the vector with $k$-th component equal to $1$ and zeros elsewhere. Since $v_i=v= D(i;v)^{(\type, n)}_i$, it clearly holds $k\neq i$ and thus, the  vector $\bv+2\be_k$ still has $v$ as its $i$-th component. Hence, we just need to prove that $\bv+2\be_k$ is a distance vector in $\cD(d+2, \type, n)$. To so so, it only remains to check condition (\ref{prop: properties distance vectors-item4}) of Theorem \ref{prop: properties distance vectors} for the $k$-th component and the adjacent ones. In other words, we must show that the relations
\begin{align}
|(v_k + 2) - v_{k-1}| & \leq  2(t_k-t_{k-1}) \ \ \text{if} \ \ 1 < k \leq r  \ \ \text{and}  \label{eq: rel 1} \\
|v_{k+1}   - (v_k+2)| & \leq  2(t_{k+1}-t_k) \ \ \text{if} \ \ 1 \leq k < r.   \label{eq: rel 2}
\end{align}
hold. We start by proving (\ref{eq: rel 1}) in case $1< k\leq r$. To do so, we distinguish two cases. First, if $v_k < v_{k-1}$, then we have 
$$
|(v_k + 2) - v_{k-1}| = v_{k-1} - v_k - 2 = |v_k - v_{k-1}| -2 \leq  2(t_k-t_{k-1}) -2 < 2(t_k-t_{k-1}).
$$
On the other hand, if $v_k \geq v_{k-1}$, by the minimality in the choice of $k$, we have that $v_{k-1}=D(i; v)^{(\type, n)}_{k-1}.$  Hence, since $\bD(i; v)^{(\type, n)}$ is a distance vector, it holds
$$
|(v_k+2)-v_{k-1}|= (v_k + 2) - v_{k-1} \leq  D(i; v)^{(\type, n)}_k - D(i; v)^{(\type, n)}_{k-1} \leq 2(t_k-t_{k-1}).
$$
The proof of (\ref{eq: rel 2}), in case that $1\leq k <r$, is completely analogous and we omit it. 
\end{proof}

To our purposes, it will be important, in turn,  to look at the behaviour of the components of distance vectors associated to a given value for the flag distance. Hence, given an even integer $0\leq d \leq D^{(\type, n)}$, we consider the values

\begin{equation}\label{def: bar d_i and bar D_i}
\bar{d}_i = \min\{ v_i \ | \ \bv\in\cD(d,\type, n) \} \ \text{and} \ \bar{D}_i = \max\{ v_i \ | \ \bv\in\cD(d,\type, n) \}.
\end{equation}
The value $\bar{d}_i$ (resp. $\bar{D}_i$) represents the minimum (resp. maximum) value that can be placed in the $i$-th component of a distance vector in $\cD(d,\type, n)$.

\begin{remark}
Notice that, the values $\bar{d}_i$ and $\bar{D}_i$ defined in (\ref{def: bar d_i and bar D_i}) satisfy the chain of inequalities 
\begin{equation}\label{eq: chain inequalities}
d(i; \bar{d}_i)^{(\type, n)} \leq d(i; \bar{D}_i)^{(\type, n)} \leq d \leq D(i; \bar{d}_i)^{(\type, n)} \leq D(i; \bar{D}_i)^{(\type, n)}.
\end{equation}
Example \ref{ex: monotonicity of dist vect} illustrates this fact.
\end{remark}

With this notation, the next result holds.

\begin{proposition}\label{prop: relation with bar d}
Let $d$ be an even integer such that $0\leq d\leq D^{(\type, n)}$. Consider an index $1\leq i\leq r$ and take an even integer $v$ with $0\leq v \leq \min\{2t_i, 2(n-t_i)\}.$ The following statements hold:
\begin{enumerate}
\item If $v< \bar{d}_i$, then we have $D(i; v)^{(\type, n)}< d$.
\item If $v> \bar{D}_i$, then $d(i; v)^{(\type, n)}> d.$
\end{enumerate}
\end{proposition}
\begin{proof}
Suppose that $v<\bar{d}_i$, then, by means of (\ref{eq: chain inequalities}), it holds 
$$
d(i; v)^{(\type, n)} \leq d(i; \bar{d}_i)^{(\type, n)} \leq d.
$$
Suppose now that $d\leq D(i; v)^{(\type, n)}$. In this case, by means of Proposition \ref{prop: every distance between dmin dmax}, there must exist a distance vector in $\cD(d,\type, n)$ with $v$ as its $i$-th component. This leads to $\bar{d}_i\leq v$, which is a contradiction. Hence, it holds $d>D(i; v)^{(\type, n)}$.

On the other hand, if $v> \bar{D}_i$, by using (\ref{eq: chain inequalities}), we clearly have that 
$$
d \leq D(i; \bar{D}_i)^{(\type, n)}\leq D(i; v)^{(\type, n)}.
$$
If we assume that $d\geq d(i;v)^{(\type, n)}$, by using Proposition \ref{prop: every distance between dmin dmax},  we can find a distance vector in $\cD(d,\type, n)$ with $v$ as its $i$-th component. This contradicts the fact that $v > \bar{D}_i$. As a result, we conclude that $d < d(i; v)^{(\type, n)}$.
\end{proof}

The previous result points out the impossibility of attaining the flag distance value $d$ when we consider subspaces distances $v$ out of the interval $[\bar{d}_i, \bar{D}_i]$ at the $i$-th summand. This fact will be useful in Section \ref{sec: bounds} and it is reflected in the next example.
\begin{example}\label{ex: monotonicity of dist vect}
As said in Example \ref{ex: not a dist vector}, the set of distance vectors associated to $d=20$ for the full flag variety on $\bbF_q^7$ is
$$
\cD(20,7)=\{(2,4,\mathbf{4},4,4,2), (2,4,\mathbf{6},4,2,2), (2,2,\mathbf{4},6,4,2) \} .
$$
Hence, for $d=20$, it is clear that $\bar{d}_3= 4$ and $\bar{D}_3=6$. Moreover, in this case expression (\ref{eq: chain inequalities}) becomes
$$
d(3;4)^7= 8 < d(3;6)^7=18  < 20 < D(3;4)^7=22 < D(3; 6)^7=24.
$$
Besides, by means of Proposition \ref{prop: relation with bar d}, the maximum distance that can be obtained by distance vectors with third component $v_3 < 4$ is lower than $20$. Indeed, that maximum distance is attained with the vector
$$
\bD(3;2)^7= (2,4,\mathbf{2},4,4,2),
$$ 
whose associated distance is $D(3;2)^7=18 < 20$.
\end{example}

\subsection{Distance vectors with prescribed zero components}

In this subsection we study the set of attainable values of the flag distance by distance vectors having prescribed zero components, i.e., distance vectors associated to pairs of flags that share certain subspaces. For the sake of simplicity, we will present partial results for the full flag variety, followed by the natural general version for $\cF_q(\type, n)$, deduced by using the projection map defined in (\ref{eq: map pi type}).

\begin{remark}
Recall that, as pointed out in Remark \ref{rem: D type n and D n}, in case of working with distance vectors with no zero components, the maximum possible distance is the value $D^{(\type,n)}$, which is attained by the vector $\bD^{(\type,n)}.$
\end{remark}

Let us start our study with distance vectors with just one zero among their components by taking advantage of the results provided in Subsection \ref{subsec: fixed component}. Later on, we generalize this and analyze the properties of distance vectors with several null components. Given the type vector $\type$ and a position $1\leq i\leq r$, by means of (\ref{eq: components d(i, a) and D(i,a)}),  it clearly holds $d(i;0)^{(\type, n)}=0$. On the other hand, now we study the value $D(i; 0)^{(\type, n)}$ and its associated distance vector $\bD(i; 0)^{(\type, n)}$. Observe that, in this case, we do not need to specify the fixed component since is always zero. Hence, we will  just write $D(i)^{(\type, n)}$ and $\bD(i)^{(\type, n)}.$ Moreover, by means of (\ref{eq: components d(i, a) and D(i,a)}), the $j$-th component $\bD(i)^{(\type, n)}$ is given by
\begin{equation}\label{eq: summands D(i) type}
D(i)^{(\type, n)}_j = \min \{ 2t_j, 2(n-t_j), 2|t_i - t_j | \}
\end{equation}
for every $1\leq j\leq r$. When working with full flags, we also drop the type vector and simply write $D(i)^n$ and $\bD(i)^n$. In this case, expression (\ref{eq: summands D(i) type}) becomes
\begin{equation}\label{eq: summands D(i)}
D(i)^n_j=\min \{ 2j, 2(n-j), 2 |i - j | \},
\end{equation}
for $1\leq j\leq n-1$. Using both (\ref{eq: summands D(i) type}) and (\ref{eq: summands D(i)}) and the map defined in (\ref{eq: map pi type}), the next result follows. 

\begin{proposition}\label{prop: proj 1 zero}
Given a type vector $\type$ on $\bbF_q^n$ and an index $1\leq i\leq r,$ it holds
$$
\bD(i)^{(\type, n)}= \pi_{\type}(\bD(t_i)^n).
$$
\end{proposition}
This fact allows us to restrict our study to the full type case. At the end of the section we will come back to the general flag variety $\cF_q(\type, n)$. 

\subsubsection*{The full type case}

We start giving some properties of the value $D(i)^n$.
\begin{proposition}\label{prop: D(i)=D(n-i)}
For every $1\leq i\leq n-1$, we have that
$$
D(i)^n_j = D(n-i)^n_{n-j}, \  \forall j=1, \dots, n-1.
$$
In other words, to obtain $\bD(n-i)^n$, it suffices to read backwards the vector $\bD(i)^n$. As a consequence, it holds $D(i)^n=D(n-i)^n$.
\end{proposition}
\begin{proof}
Take an integer $1\leq i\leq n-1$. According to (\ref{eq: summands D(i)}),  for every $1\leq j\leq n-1,$ it clearly holds
$$
\begin{array}{rcl}
D(n-i)^n_{n-j} &=& \min \{ 2(n-j), 2(n-(n-j)), 2 |(n-i) -(n-j)| \}\\
               &=& \min \{2(n-j), 2j, 2 |j-i| \} = D(i)^n_j,
\end{array}
$$
which gives the result straightforwardly. 
\end{proof}

In light of this result, we just need to study the values $D(i)^n$ for $1\leq i\leq \lfloor \frac{n}{2}\rfloor$. We can also give the following nice description of $D(i)^n$ in terms of the values $D^j$ defined in (\ref{eq: D^n}).

\begin{proposition}\label{prop: D(i)^n = D^i + D^{n-i}}
For every $1\leq i\leq n-1$, it holds
$$
D(i)^n = D^i + D^{n-i}.
$$
\end{proposition}
\begin{proof}
Regarding equation (\ref{eq: summands D(i)}), we can compute the value $D(i)^n$ as
$$
D(i)^n = \sum_{j=1}^{n-1} \min\{ 2j, 2(n-j), 2|i-j|\}.
$$

Observe that, in case $i=1$, we have
$$
\begin{array}{cll}
D(1)^n &=& \sum_{j=2}^{n-1} \min\{ 2j, 2(n-j), 2(j-1)\}  \\
       &=& \sum_{j=2}^{n-1} \min\{2(j-1), 2(n-j)\} \\
       &=& \sum_{k=1}^{n-2} \min\{2k, 2((n-1)-k)\}\\
       &=& D^{n-1} = D^1+D^{n-1}.
\end{array}
$$
Besides, by means of Proposition \ref{prop: D(i)=D(n-i)}, the result also holds if $i=n-1$. Let us now consider the case $1 < i< n-1$. In this case, the $i$-th component is zero and we have
\begin{equation}\label{eq: separate summands D(i)}
D(i)^n = \sum_{j=1}^{i-1} \min\{ 2j, 2(n-j), 2|i-j|\} + 0 + \sum_{j=i+1}^{n-1} \min\{ 2j, 2(n-j), 2|i-j|\}.
\end{equation}
Moreover, for values of $j<i$, one have that $2|i-j|=2(i-j)<2(n-j)$. On the other hand, if $i<j$, it is clear that $2|i-j|= 2(j-i) < 2j.$ Hence, (\ref{eq: separate summands D(i)}) becomes
$$
\begin{array}{ccl}
D(i)^n &=& \sum_{j=1}^{i-1} \min\{ 2j, 2(i-j)\} + \sum_{j= i+1}^{n-1} \min\{2(n-j), 2(j-i)\} \\
      &=& \sum_{j=1}^{i-1} \min\{ 2j, 2(i-j)\} + \sum_{k=1}^{n-i-1} \min\{2(n-i-k), 2k\}\\
      &=& D^i + D^{n-i},
\end{array}
$$where the second equality comes from writing $k=j-i$.
\end{proof}

This result confirms again the fact that $D(i)^n=D(n-i)^n$. Next, we use the previous proposition together with expression (\ref{eq: D^n}) to provide an explicit formula for every $D(i)^n$.

\begin{corollary}\label{cor: explicit formula D(i)}
For every $1\leq i\leq n-1$, it holds
$$
D(i)^n =
\left\lbrace
\begin{array}{cl}
\dfrac{i^2+(n-i)^2}{2} & \text{if both} \ n \ \text{and } \ i \ \text{are even,}\\
\dfrac{i^2+(n-i)^2-2}{2} & \text{if} \ n \ \text{is even and} \ i \ \text{is odd,}\\
\dfrac{i^2+(n-i)^2-1}{2} & \text{if} \ n \ \text{is odd}.
\end{array}
\right.
$$
\end{corollary}

This expression allows us to establish an order on the set $\{D(i)^n  \ | \ 1\leq i\leq \lfloor \frac{n}{2}\rfloor\}$.

\begin{proposition}\label{prop: Order on the D(i)^n}
For every $n$ it holds
$$
D^n > D(1)^n > D(2)^n > \dots > D(\lfloor\frac{n}{2}\rfloor-1)^n \geq D(\lfloor\frac{n}{2}\rfloor)^n,
$$
and the last equality holds if, and only if, $4$ divides $n$.
\end{proposition}
 
\begin{proof}
Consider any $1\leq i< \lfloor\frac{n}{2}\rfloor$. We will use the expression in Corollary \ref{cor: explicit formula D(i)} in order to compare $D(i)^n$ and $D(i+1)^n$. We do so by dividing the proof into two parts, depending on the parity of $n$. First of all, assume that $n$ is odd. In this case, it follows $i+1\leq \lfloor \frac{n}{2}\rfloor = \frac{n-1}{2}$. Then $2i+1 < 2i+3 \leq n$ and
$$
\begin{array}{lll}
D(i+1)^n & = & \frac{(i+1)^2+(n-(i+1))^2-1}{2} \\
             & = & \frac{i^2+2i+1+(n-i)^2-2(n-i)+1-1}{2}\\
             & = & \frac{i^2+(n-i)^2-1}{2}+\frac{2(2i+1-n)}{2}\\
             & = & D(i)^n -(n-(2i+1)) < D(i)^n.
\end{array}
$$

Now, suppose that $n$ is an even integer and $1\leq i < \lfloor\frac{n}{2}\rfloor = \frac{n}{2}$. Equivalently, $2(i+1)\leq n$. We distinguish two cases:
\begin{itemize}
\item if $i$ is even, then $i+1$ is odd and we have
$$
\begin{array}{lll}
D(i+1)^n & = & \frac{(i+1)^2+(n-(i+1))^2-2}{2} \\
             & = & \frac{i^2+2i+1+(n-i)^2-2(n-i)+1-2}{2}\\
             & = & \frac{i^2+(n-i)^2}{2}+\frac{2(2i-n)}{2}\\
             & = & D(i)^n - (n-2i) < D(i)^n.
\end{array}
$$

\item On the other hand, if $i$ is odd, then $i+1$ is even and:
$$
\begin{array}{lll}
D(i+1)^n & = & \frac{(i+1)^2+(n-(i+1))^2}{2} \\
             & = & \frac{i^2+2i+1+(n-i)^2-2(n-i)+1}{2}\\
             & = & \frac{i^2+(n-i)^2-2}{2}+\frac{2(2(i+1)-n)}{2}\\
             & = & D(i)^n -(n-2(i+1)) \leq  D(i)^n
\end{array}
$$and the last equality holds if, and only if, $n=2(i+1)$ and then $4$ divides $n$.
\end{itemize}
\end{proof}

The next example reflects the information given in the previous results for a specific value of $n$.
\begin{example}
For $n=7$, we have $D^7=24$. In this example, we compute all the values $D(i)^7$ and respective vectors $\bD(i)^7$.
\begin{table}[H]
\begin{center}
\begin{tabular}{ccccccc}
\cline{1-3} \cline{5-7}
$i$ & $\bD(i)^7$      & $D(i)^7$ &  & $i$ & $\bD(i)^7$      & $D(i)^7$ \\ \cline{1-3} \cline{5-7} 
$1$ & $(0,2,4,6,4,2)$ &    18    &  & $6$ & $(2,4,6,4,2,0)$ &   18     \\ \cline{1-3} \cline{5-7} 
$2$ & $(2,0,2,4,4,2)$ &    14    &  & $5$ & $(2,4,4,2,0,2)$ &   14     \\ \cline{1-3} \cline{5-7} 
$3$ & $(2,2,0,2,4,2)$ &    12    &  & $4$ & $(2,4,2,0,2,2)$ &   12     \\ \cline{1-3} \cline{5-7} 
\end{tabular}
\end{center}
\end{table}
Notice that, as stated in Proposition \ref{prop: D(i)=D(n-i)}, for every $1\leq i\leq 6$, it holds $D(i)^7=D(7-i)^7$. Moreover, the reader can see that every vector $\bD(i)^7$ has the same components than $\bD(n-i)^7$ but written backwards. Moreover, it is also shown that 
$$
D^7 > D(1)^7 > D(2)^7 > D(3)^7.
$$
We will come back to this example in Section \ref{sec: example}.
\end{example}

Recall that the flag distance between full flags on $\bbF_q^n$ is an even integer in the interval $[0, D^n]$. Moreover, in light of Proposition \ref{prop: Order on the D(i)^n}, we can partition this interval into intervals of the form $]D(i+1)^n, D(i)^n]$ that will be used in Sections \ref{sec: disjointness} and \ref{sec: bounds} to obtain useful information about full flag codes.

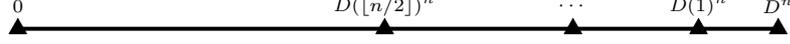
\begin{figure}[H]
\begin{center}
\begin{tikzpicture}[line cap=round,line join=round,>=triangle 45,x=1cm,y=1cm]
\draw [line width=1.5pt] (0,0)-- (10,0);
\begin{scriptsize}
\draw [fill=black,shift={(0,0)}] (0,0) ++(0 pt,3.75pt) -- ++(3.2475952641916446pt,-5.625pt)--++(-6.495190528383289pt,0 pt) -- ++(3.2475952641916446pt,5.625pt);
\draw[color=black] (0,0.3) node {$0$};
\draw [fill=black,shift={(10,0)}] (0,0) ++(0 pt,3.75pt) -- ++(3.2475952641916446pt,-5.625pt)--++(-6.495190528383289pt,0 pt) -- ++(3.2475952641916446pt,5.625pt);
\draw[color=black] (10,0.3) node {$D^n$};
\draw [fill=black,shift={(4.8245463484680045,0)}] (0,0) ++(0 pt,3.75pt) -- ++(3.2475952641916446pt,-5.625pt)--++(-6.495190528383289pt,0 pt) -- ++(3.2475952641916446pt,5.625pt);
\draw[color=black] (4.8245463484680045,0.3) node {$D(\lfloor n/2\rfloor)^n$};
\draw [fill=black,shift={(7.301092033683474,0)}] (0,0) ++(0 pt,3.75pt) -- ++(3.2475952641916446pt,-5.625pt)--++(-6.495190528383289pt,0 pt) -- ++(3.2475952641916446pt,5.625pt);
\draw[color=black] (7.301092033683474,0.3) node {$\cdots$};
\draw [fill=black,shift={(8.952122490493787,0)}] (0,0) ++(0 pt,3.75pt) -- ++(3.2475952641916446pt,-5.625pt)--++(-6.495190528383289pt,0 pt) -- ++(3.2475952641916446pt,5.625pt);
\draw[color=black] (8.952122490493787,0.3) node {$D(1)^n$};
\end{scriptsize}
\end{tikzpicture}
\caption{Distribution of the values $D(i)^n$.}
\end{center}
\end{figure}

In order to give a partition of the left interval $[0, D(\lfloor n/2\rfloor)^n]$, let us introduce another set of relevant distances that correspond to the maximum possible flag distances associated to distance vectors with more that one zero among their components.

\begin{definition}\label{def: D(i1,...iM) full}
Consider an integer value $0\leq M\leq n-1$ and fix dimensions $1\leq i_1 < i_2 < \dots < i_M \leq n-1$. We write $D(i_1, \dots, i_M)^n$ to denote the maximum possible distance attainable by distance vectors in $\cD(n)$ with $M$ zeros in the positions $i_1, \dots, i_M$. This situation corresponds uniquely to the distance vector $\bD(i_1, \dots, i_M)^n$, whose $j$-th component is given by
\begin{equation}\label{eq: components bD(i_1,...,i_M)n}
D(i_1, \dots, i_M)^n_j=\min \{ 2j, 2(n-j), 2|j-i_1|, \dots, 2|j-i_M| \},
\end{equation}
for all $1\leq j\leq n-1$. 
\end{definition}

Notice that, in case $M=0$, we have the distance vector $\bD^n$ given in Remark \ref{rem: D type n and D n}. The case $M=1$ corresponds to the distance vector $\bD(i)^n$ defined in (\ref{eq: summands D(i)}). On the other hand, if $M=n-1$, then $D(i_1, \dots, i_M)^n=0$ and $\bD(i_1, \dots, i_M)^n$ is the null vector.

\begin{proposition}\label{prop: separate summands D(i_1,...,i_M)n}
Given indices $1\leq i_1 < \dots < i_M \leq n-1$, we have
$$
D(i_1, \dots, i_M)^n = D^{i_1} + D^{i_2-i_1} + \dots + D^{i_M-i_{M-1}}+ D^{n-i_M}.
$$
\end{proposition}
\begin{proof}

We prove the result by induction on the number of zeros $1\leq M\leq n-1$. Observe that, by means of Proposition \ref{prop: D(i)^n = D^i + D^{n-i}}, the result holds for every $n$ and $M=1$. Now, assume that $M>1$ and, by induction hypothesis, that the result is true for all $n$ and distance vectors having up to $M-1$ zeros. Let us study the case of having $M$ zeros in the positions $i_1, \dots, i_M$. We start by considering the case where $i_1=1$. In this situation, according to (\ref{eq: components bD(i_1,...,i_M)n}), 
\begin{small}
$$
\begin{array}{ll}     
D(1, i_2, \dots, i_M)^n  & = \sum\limits_{j=1}^{n-1} \min \{ 2j, 2(n-j), 2|j-1|, 2|j-i_2|, \dots, 2|j-i_M| \} \\
            			 & =  0 + \sum\limits_{j=2}^{n-1} \min \{ 2j, 2(n-j), 2|j-1|, 2|j-i_2|, \dots, 2|j-i_M| \} \\
			             & =  0 + \sum\limits_{k=1}^{n-2} \min \{2((n-1)-k), 2k, 2|k-(i_2-1)|, \dots, 2|k-(i_M-1)| \}\\
			             & \\ [-1em]
		                 & =  D^1 + D(i_2-1, \dots, i_M-1)^{n-1},
\end{array}
$$
\end{small}

\noindent where the third equality comes from taking $k=j-1$. Hence, the induction hypothesis leads to 
$$
D(1, i_2, \dots, i_M)^n = D^1 +D^{i_2-1} + D^{i_3-i_2} + \dots + D^{i_M-i_{M-1}} + D^{n-i_M},
$$
as stated. Assume now that $i_1 > 1$, then we obtain 
$$
\begin{array}{ccl}
D(i_1, \dots, i_M)^n & = & \sum_{j=1}^{n-1} \min \{ 2j, 2(n-j), 2|j-i_1|, \dots, 2|j-i_M| \}\\
                     & = & \sum_{j=1}^{i_1-1} \min \{ 2j, 2(n-j), 2|j-i_1|, \dots, 2|j-i_M| \} + 0 \\
                     &   & + \sum_{j=i_1+1}^{n-1} \min \{ 2j, 2(n-j), 2|j-i_1|, \dots, 2|j-i_M| \}.
\end{array}
$$
Observe that, when $1\leq j\leq i_1-1$, it holds $j< i_1 <i_2< \dots < i_M< n$ and then $n-j > |j-i_l|=i_l-j \geq j-i_1$ for all $1\leq l\leq M$. Hence, the first part of the last sum can be substituted by 
$$
\sum_{j=1}^{i_1-1} \min \{ 2j, 2(i_1-j)\} = D^{i_1}.
$$
On the other hand, if $i_1+1\leq j\leq n-1$, clearly $|j-i_1|= j-i_1\leq j$ and then
\begin{small}
$$
\begin{array}{ll}
D(i_1, \dots, i_M)^n \hspace{-0.4cm}  & = D^{i_1} + \sum\limits_{j=i_1+1}^{n-1} \min \{2(n-j), 2(j-i_1) , 2|j-i_2|, \dots, 2|j-i_M| \} \\
					  & = D^{i_1} + \sum\limits_{k=1}^{n-i_1-1} \min \{ 2(n-i_1-k), 2k, 2|k-(i_2-i_1)|, \dots, 2|k-(i_M-i_1)| \} \\
					  & \\ [-1em]
					  & = D^{i_1} + D(i_2-i_1, \dots, i_M-i_1)^{n-i_1}.
\end{array}
$$
\end{small}

Hence, by applying the induction hypothesis to $D(i_2-i_1, \dots, i_M-i_1)^{n-i_1}$, we obtain
$$
\begin{array}{ccl}
D(i_1, \dots, i_M)^n & = & D^{i_1} + D^{i_2-i_1} + \dots  + D^{(i_M-i_1)-(i_{M-1}-i_1)} + D^{(n-i_1)-(i_{M}-i_1)}\\
                     & = & D^{i_1} + D^{i_2-i_1} + \dots  + D^{i_M-i_{M-1}} + D^{n-i_{M}},
\end{array}
$$
as we wanted to prove.
\end{proof}

Proposition \ref{prop: separate summands D(i_1,...,i_M)n} along with expression (\ref{eq: D^n}) allows us to compute every value $D(i_1, \dots, i_M)^n$ depending on the parity of the positive integers $i_1, i_2-i_1, \dots, n-i_M$, as we see in the following example.

\begin{example}
For $n=7$, $M=3$ and indices $i_1=1$, $i_2=3$ and $i_3=4$, by means of Proposition \ref{prop: separate summands D(i_1,...,i_M)n}, we have
$$
\begin{array}{ccl}
D(1,3,4)^7 &=& D^1 + D^{3-1} + D^{4-3} + D^{7-4} \\
           &=& D^1 + D^2 + D^1 + D^3             \\
           &=& 0 + \frac{2^2}{2} + 0 + \frac{3^2-1}{2}  = 6.
\end{array}
$$
\end{example}

\begin{remark}\label{rem: ordered differences}
A \emph{multiset} is a collection whose elements can appear more than once. The number of times that each element appears in the multiset is its \emph{multiplicity}. We represent multisets by using double braces $\{\{ \ldots \}\}$.  Notice that, for any two families of $M$ indices $1\leq i_1 < \dots < i_M\leq n-1$ and $1\leq j_1 < \dots < j_M\leq n-1$ satisfying the equality of multisets
$$ 
\{\{  i_1, i_2-i_1, \dots, n-i_M \}\} = \{\{ j_1, j_2-j_1, \dots, n-j_M \}\},
$$
by means of Proposition \ref{prop: separate summands D(i_1,...,i_M)n}, we have the equality $D(i_1, \dots, i_M)^n=D(j_1, \dots, j_M)^n.$ Hence, in order to compute all the values $D(i_1, \dots, i_M)^n$, we can restrict ourselves to choices of $M$ ordered indices $1\leq i_1 < \dots < i_M\leq n-1$ such that the differences
$$
1\leq i_1 \leq i_2-i_1 \leq \dots \leq i_M-i_{M-1} \leq n-M
$$
are also ordered.
In general, the converse is not true: there are different families of multisets as above providing the same values of the distance. It suffices to see that, with $M=1$ and $n=8$, it holds $D(3)^8=16=D(4)^8$. However, the multisets of differences associated to indices $i=3$ and $i=4$ are $\{\{3, 5\}\}$ and $\{\{4, 4\}\}$, respectively.
\end{remark}

The next result establishes that the maximum distance attainable by distance vectors in $\cD(n)$ with $M$ zero components is always obtained when these zeros are placed in the first $M$ positions. 

\begin{proposition}\label{prop: D(1,...M) is the largest}
Given $1\leq M\leq n-1$ and any election of indices $1\leq i_1 < \dots < i_M \leq n-1$, it holds
$$
D(i_1, \dots, i_M)^n \leq D(1, \dots, M)^n = D^{n-M}.
$$
\end{proposition}

\begin{proof}
Notice that $D(1, \dots, M)^n = D^{n-M}$ holds by application of Proposition \ref{prop: separate summands D(i_1,...,i_M)n}. Hence, we just  need to prove the first inequality. To do so, we proceed by induction on $M$. We start with the case $M=1$, in which, by means of Proposition \ref{prop: Order on the D(i)^n}, it is clear that $D(n-i)^n=D(i)^n\leq D(1)^n$, for every value of $n$ and $1\leq i\leq n-1$.

Assume now that $M>1$ and that the result holds for any value of $n$ and distance vectors in $\cD(n)$ having up to $M-1$ ceros. Let us prove that it is also true for $M$ zeros. To do so, consider $M$ indices $1\leq i_1< \dots < i_M\leq n-1$. Notice that $M\leq i_M$. Moreover, if $M=i_M$, then it  holds $i_j=j$, for every $j\in\{1,\dots,M\}$ and therefore $D(i_1,\dots,i_M)^n=D(1,\dots,M)^n.$ Assume now that $M< i_M$. By means of Proposition \ref{prop: separate summands D(i_1,...,i_M)n} , we have
$$
\begin{array}{ccl}
D(i_1, \dots, i_M)^n & = & D^{i_1}+ D^{i_2-i_1} +  \dots + D^{i_M-i_{M-1}} + D^{n-i_M}\\
                     & = & D(i_1, \dots, i_{M-1})^{i_M} + D^{n-i_M}.
\end{array}
$$
Moreover, since $M-1 < i_M-1$, we can apply the induction hypothesis to the case of having $M-1$ zeros in the positions $i_1 < \dots < i_{M-1}$ on $\bbF_q^{i_M}.$ We obtain $D(i_1, \dots, i_{M-1})^{i_M}\leq D^{i_M-(M-1)}$ and Proposition \ref{prop: D(i)^n = D^i + D^{n-i}} along with Proposition \ref{prop: Order on the D(i)^n} gives that 
$$
\begin{array}{ccl}
D(i_1, \dots, i_M)^n & \leq &  D^{i_M-(M-1)} + D^{n-i_M}\\
                     & =    &  D(i_M-M+1)^{i_M-M+1+n-i_M}\\
                     & =    &  D(i_M-M+1)^{n-M+1}\\
                     & \leq &  D(1)^{n-M+1} = D^1 + D^{n-M+1-1} = D^{n-M},
\end{array}
$$ 
as we wanted to prove.

\end{proof}

\subsubsection*{The general case}
We finish this section by generalizing the previous concepts to the general flag variety $\cF_q(\type, n)$ as follows.  As done for the full type case in Definition \ref{def: D(i1,...iM) full}, we can consider distance vectors $\cD(\type, n)$ with a prescribed number of zeros, say $0\leq M\leq r$, in the positions $1\leq i_1 < \dots < i_M \leq r$. We denote the corresponding maximum distance by
$$
D(i_1, \dots, i_M)^{(\type, n)}.
$$
This number represents the maximum possible distance between flags in $\cF_q(\type, n)$ that share simultaneously their subspaces of dimensions $t_{i_1}, \dots, t_{i_M}$. The only distance vector giving this distance and having zeros in its components $i_1, \dots, i_M$ is denoted by $\bD(i_1, \dots, i_M)^{(\type, n)}$ and its $j$-th component is given by
\begin{equation}\label{eq: components bD(i1,..., iM)(type, n)}
D(i_1, \dots, i_M)^{(\type, n)}_j = \min \{ 2t_j, 2(n-t_j), 2|t_j-t_{i_1}|, \dots,  2|t_j-t_{i_M}| \},
\end{equation}
for $1\leq j\leq r$.

Using the projection map $\pi_{\type}$ defined in (\ref{eq: map pi type}), we can give the following description of the distance vector with components as in (\ref{eq: components bD(i1,..., iM)(type, n)}), in terms of  the vector $\bD(t_{i_1}, \dots, t_{i_M})^n$ introduced in Definition \ref{def: D(i1,...iM) full}. The next result generalizes Proposition \ref{prop: proj 1 zero}.

\begin{proposition}\label{prop: projection dist vectors many zeros}
Given a type vector $\type$ and $1\leq M\leq r$ indices $1\leq i_1 < \dots < i_M \leq r$, it holds
$$
\bD(i_1, \dots, i_M)^{(\type, n)} = \pi_{\type}(\bD(t_{i_1}, \dots, t_{i_M})^n).
$$
\end{proposition}
\begin{proof}
Consider the type vector $\type$ and take $1\leq M\leq r$ indices $1\leq i_1 < \dots < i_M \leq r$. Notice that, for every $1\leq j\leq r$, the $j$-th component of $\pi_{\type}(\bD(t_{i_1}, \dots, t_{i_M})^n)$ is exactly the $t_j$-th one of  $\bD(t_{i_1}, \dots, t_{i_M})^n$ that, by (\ref{eq: components bD(i_1,...,i_M)n}), is:
$$
\min \{ 2t_j, 2(n-t_j), 2|t_j-t_{i_1}|, \dots, 2|t_j-t_{i_M}| \}.
$$
This value corresponds to the $j$-th component of $\bD(i_1, \dots, i_M)^{(\type, n)}$, as we wanted to prove.
\end{proof}

Next, we give a generalization of Proposition \ref{prop: separate summands D(i_1,...,i_M)n} for any arbitrary type vector $\type=(t_1, \dots, t_r)$. To do so, consider $1\leq M\leq r$ zeros in the positions $1\leq i_1 < \dots < i_M \leq r$. These positions allow us to split $\type$ into $M+1$ new type vectors, that we denote by $\type^1, \dots, \type^{M+1}$, given by
\begin{equation}\label{def: subtypes}
\left\lbrace
\begin{array}{ccl}
\type^1     &=& (t_1, \dots, t_{i_1-1}),  \\ 
\type^{j+1} &=& (t_{i_{j}+1} - t_{i_{j}}, \dots, t_{i_{j+1}-1}-t_{i_{j}}), \ \text{for} \ 1\leq j\leq M-1, \\
\type^{M+1} &=& (t_{i_{M}+1}-t_{i_{M}}, \dots, t_{r}-t_{i_{M}}). 
\end{array}
\right.
\end{equation}

Using this notation, the next result holds.
\begin{proposition}\label{prop: separate summands general type}
Given a type vector $\type$ and a choice of $1\leq M\leq r$ ordered indices $1\leq i_1< \dots <i_M\leq r$, then the value $D(i_1, \dots, i_M)^{(\type, n)}$ satisfies:
$$
D(i_1, \dots, i_M)^{(\type, n)}  =  D^{(\type^1, t_{i_1})}+ D^{(\type^2, t_{i_2}-t_{i_1})} +\dots + D^{(\type^{M+1}, n-t_{i_M})}.
$$
\end{proposition} 

The next example reflects this fact.

\begin{example}
Take $n=12$ and consider the type vector $\type= (1,3,5,6,8,10,11)$ of length $r=7$. Assume that we place $M=2$ zeros in the positions $i_1=3$ and $i_2=5$, i.e., the ones corresponding to the dimensions $t_3=5$ and $t_5=8$. In this case, by means of (\ref{def: subtypes}), we have
$$
\type^1=(1, 3), \ \ \type^2= (6-5)= (1) \ \ \text{and} \ \ \type^3= (10-8, 11-8)= (2,3).
$$
Moreover, by (\ref{eq: components bD(i1,..., iM)(type, n)}), it holds
$$
\bD(3,5)^{(\type, 12)} = ( 2, 4, \mathbf{0}, 2, \mathbf{0}, 4, 2).
$$
Observe that the zero components of $\bD(3,5)^{(\type, 12)}$ allow us to split this vector into three new ones, which are precisely
$$
\bD^{(\type^1, 5)} = (2,4), \ \ \bD^{(\type^2, 8-5)}=(2) \ \ \text{and} \ \  \bD^{(\type^3, 12-8)} =(4,2).
$$
Hence, we have
$$
D(3,5)^{(\type, 12)}= 2 + 4 + 0 + 2 + 0+ 4  + 2  = D^{(\type^1, 5)} + D^{(\type^2, 8-5)} + D^{(\type^3, 12-8)},
$$
as stated in Proposition \ref{prop: separate summands general type}.
\end{example}

\begin{remark}\label{rem: values Ds only depend on the variety}
Notice that the computation of the distance $D(i_1, \dots, i_M)^{(\type, n)}$ only depends on the flag variety $\cF_q(\type, n)$ and on the choice of the indices $i_1, \dots, i_M$. As a result, these values can be computed in advance, before considering any particular flag code, as we will see in Section \ref{sec: example}.
\end{remark}

In the following sections we will take advantage of this study of the values $D(i_1, \dots, i_M)^{(\type, n)}$ in order to derive some properties related to the structure and cardinality of flag codes.

\section{Disjointness in flag codes}\label{sec: disjointness} 

Recall that, given a flag code $\cC\subseteq \cF_q(\type, n)$, for every $1\leq i\leq r$, its \emph{$i$-th projected code} is the constant dimension code 
$$
\cC_i = p_i(\cC) \subseteq \cG_q(t_i, n),
$$
where $p_i$ is the projection map defined in (\ref{def: projection}). As a consequence, for every $1\leq i \leq r$, we have
$$
|\cC_i|= |p_i(\cC)| \leq |\cC|
$$
and the equality holds if, and only if, the projection $p_i$ is injective when restricted to $\cC$. If we have the equality for all $1\leq i\leq r$, i.e., if $|\cC|=|\cC_1|=\dots =|\cC_r|$, the flag code $\cC$ is said to be \emph{disjoint} (see  \cite{CasoPlanar}). Under the \emph{disjointness} property, the code cardinality is completely determined by its projected codes and different flags never share a subspace. Moreover, observe that every flag code $\cC\subseteq\cF_q(\type, n)$ with $|\cC|=1$ is trivially disjoint and it holds $d_f(\cC)= \sum_{i=1}^r  d_S(\cC_i)=d_S(\cC_i)=0$, for every $1\leq i\leq r$. On the other hand, if $\cC$ is a disjoint flag code with $|\cC|\geq 2$, then 
\begin{equation}\label{eq: bound distance disjoint}
d_f(\cC) \geq \sum_{i=1}^r  d_S(\cC_i)
\end{equation}
and $d_S(\cC_i)>0$, for every $1\leq i\leq r$. In particular, we obtain $d_f(\cC) \geq 2r$.

\begin{remark}
Disjoint flag codes in $\cF_q(\type, n)$ in which expression (\ref{eq: bound distance disjoint}) holds with equality are called \emph{consistent} (see  \cite{Consistent}). It is quite easy to see that this family of disjoint flag codes is also characterized by the property of having as a unique distance vector $(d_S(\cC_1), \dots, d_S(\cC_r))$. Optimum distance flag codes in $\cF_q(\type, n)$ are a particular class of consistent flag codes whose associated distance vector is $\bD^{(\type, n)}$ defined in (\ref{eq: components vector D type n}).
\end{remark}

The simple structure of disjoint flag codes leads us to seek a generalization of this concept. We do so by using the next family of projections. Consider the flag variety $\cF_q(\type, n)$ and take $1\leq M\leq r$ indices $1\leq i_1 < i_1 < \dots < i_M\leq r$. The \emph{$(i_1, \dots, i_M)$-projection map} is given as
\begin{equation}\label{def: general projection}
\begin{array}{ccccc}
p_{(i_1, \dots, i_M)}  & : & \cF_q(\type, n)        & \longrightarrow & \cF_q((t_{i_1}, \dots, t_{i_M}), n) \\
                       &   &  (\cF_1, \dots, \cF_r) & \longmapsto     & (\cF_{i_1}, \dots, \cF_{i_M})
\end{array}
\end{equation}
and the value $M$ will be called the \emph{length of the projection}. Now, given a flag code $\cC$ in $\cF_q(\type,n)$, we can define a set of flag codes of length $M$, naturally associated to $\cC$, by using these projection maps. 
\begin{definition}
Let $\cC\subset \cF_q(\type, n)$ be a flag code and fix $1\leq M\leq r$ indices $1\leq i_1 < \dots < i_M\leq r$.  The set $p_{(i_1, \dots, i_M)}(\cC)$ is called \emph{the $(i_1, \dots, i_M)$-projected code of $\cC$}. The images of $\cC$ by all the projections of length $M$ constitute the set of the so-called \emph{projected codes of length $M$ of $\cC$}.
\end{definition}

Observe that in case $M=1$, both projections $p_{i_1}$ and $p_{(i_1)}$, defined in  (\ref{def: projection}) and (\ref{def: general projection}) respectively, coincide. Hence, the $(i)$-projected code is just the $i$-projected (subspace) code defined in Section \ref{sec: prelim}, seen now as a flag code of length one.  

Next, we use these new projected codes and we introduce two wider notions of disjointness.

\begin{definition}
Let $\cC\subseteq\cF_q(\type, n)$ be a flag code and take $1\leq M\leq r$ specific indices $1\leq i_1 < \dots < i_M\leq r$. The code $\cC$ is said to be \emph{$(i_1, \dots, i_M)$-disjoint} if the projection $p_{(i_1, \dots, i_M)}$ is injective when restricted to $\cC$. If this condition holds for every choice of $M$ indices $1\leq i_1 < \dots < i_M\leq r$, we say that $\cC$ is \emph{$M$-disjoint}.
\end{definition}
 
According to this definition, we provide the next geometric interpretation of $(i_1, \dots, i_M)$-disjoint flag codes.
\begin{remark}
Consider the type vector $\type$ and $1\leq M\leq r$ indices $1\leq i_1 < \dots < i_M\leq r$. A code $\cC\subseteq\cF_q(\type, n)$ is $(i_1, \dots, i_M)$-disjoint if, and only if, different flags in $\cC$ never share simultaneously their subspaces of dimensions $t_{i_1}, \dots, t_{i_M}$. Similarly, $\cC$ is $M$-disjoint if different flags $\cC$ never have $M$ equal subspaces.
\end{remark}

\begin{example}
Let $\{\be_1, \be_2, \be_3, \be_4, \be_5 \}$ be the standard $\bbF_q$-basis of $\bbF_q^5$. We consider the full flag code $\cC$ on $\bbF_q^5$ given by the flags
$$
\begin{array}{cccccc} 
\mathcal{F}^1 &=& (\left\langle \be_1 \right\rangle, & \left\langle \be_1, \be_2 \right\rangle , & \left\langle \be_1, \be_2, \be_3 \right\rangle,  &\left\langle \be_1, \be_2, \be_3, \be_4 \right\rangle),\\
\mathcal{F}^2 &=& (\left\langle \be_1 \right\rangle, & \left\langle \be_1, \be_3 \right\rangle , & \left\langle \be_1, \be_2, \be_3 \right\rangle, &\left\langle \be_1, \be_2, \be_3, \be_4 \right\rangle),\\
\mathcal{F}^1 &=& (\left\langle \be_1 \right\rangle, & \left\langle \be_1, \be_3 \right\rangle , & \left\langle \be_1, \be_3, \be_5 \right\rangle, & \left\langle \be_1, \be_3, \be_4, \be_5 \right\rangle).
\end{array}
$$
On the one hand, observe that no pair of flags in $\cC$ share their second and third subspaces at the same time, i.e., $\cC$ is a $(2,3)$-disjoint flag code. On the other hand, it is not $(i_1, i_2)$-disjoint for any other choice of indices $1\leq i_1 < i_2\leq 4$. As a result, the code $\cC$ is not $2$-disjoint.
\end{example}

\begin{proposition}
Let $\cC\subseteq \cF_q(\type, n)$ be an $(i_1, \dots, i_M)$-disjoint flag code for some choice of $1\leq M\leq r$ indices $1\leq i_1 < \dots < i_M\leq r$. Then, for every choice of $M\leq N\leq r$ integers $1\leq j_1 < \dots < j_N\leq r$ such that $\{i_1, \dots, i_M\}\subseteq \{j_1, \dots, j_N\}$, the code $\cC$ is $(j_1, \dots, j_N)$-disjoint. In particular, if $\cC$ is $M$-disjoint, then it is $N$-disjoint as well.
\end{proposition}
\begin{proof}
Assume that $\cC$ is not a $(j_1, \dots, j_N)$-disjoint flag code. Hence, there exist different flags $\cF, \cF\in\cC$ such that 
$(\cF_{j_1}, \dots, \cF_{j_N})=(\cF'_{j_1}, \dots, \cF'_{j_N})$. Since $\{i_1, \dots, i_M\}\subseteq \{j_1, \dots, j_N\}$, then we have $(\cF_{i_1}, \dots, \cF_{i_M})=(\cF'_{i_1}, \dots, \cF'_{i_M})$, which is a contradiction with the fact that $\cC$ is an $(i_1, \dots, i_M)$-disjoint flag code. Similarly, assume now that $\cC$ is not $N$-disjoint. The previous argument leads to different flags sharing $N\geq M$ subspaces at the same time. In other words, the code $\cC$ cannot be $M$-disjoint. 
\end{proof}

At this point, we relate the $M$-disjointness property of a flag code with its minimum distance. These relationships will help us to establish bounds for flag codes in Section \ref{sec: bounds}. We start giving a lower bound for the distance of $M$-disjoint flag codes in terms of the distances of some of their projected codes of length $1$.
\begin{proposition}\label{prop: lower bound distance M-disjoint}
Let $\cC\subseteq\cF_q(\type, n)$  be a flag code and consider an integer $1\leq M\leq r.$ If  $\cC$ is $M$-disjoint, then there exist $r-(M-1)$ indices $1\leq i_1 < \dots < i_{r-M+1}\leq r$ such that  $d_S(\cC_{i_j})\neq 0$ and
$$
d_f(\cC)\geq \sum_{j=1}^{r-M+1} d_S(\cC_{i_j}).
$$  
\end{proposition}
\begin{proof}
Let $\cC\subset \cF_q(\type, n)$ be an $M$-disjoint flag code for some integer $1\leq M \leq r$ and consider a pair of different flags $\cF, \cF'\in\cC$ giving the minimum distance. The $M$-disjointness condition makes that $\cF$ and $\cF'$ cannot share more than $M-1$ subspaces. Hence, their  associated distance vector, i.e., the vector
$$
\bd(\cF, \cF')=(d_S(\cF_1, \cF'_1), \dots, d_S(\cF_r, \cF'_r))
$$ 
does not contain more than $M-1$ zeros. As a result, at least, $r-(M-1)$ of its $r$ components are nonzero. Thus, there exist different indices $1\leq i_1 < \dots < i_{r-M+1}\leq r$ such that $d_S(\cF_{i_j}, \cF'_{i_j})\neq 0$. Consequently, we have $d_S(\cF_{i_j}, \cF'_{i_j})\geq d_S(\cC_{i_j})> 0$ and then
$$
d_f(\cC)=d_f(\cF, \cF') \geq \sum_{j=1}^{r-(M-1)} d_S(\cF_{i_j}, \cF'_{i_j}) \geq \ \sum_{j=1}^{r-(M-1)} d_S(\cC_{i_j}).
$$
In other words, the distance of $\cC$ is lower bounded by the sum of nonzero distances of $r-(M-1)$ specific projected codes of length $1$. 
\end{proof}

Observe that, in the previous proof, the choice of the $r-(M-1)$ indices $1\leq i_1 < \dots < i_{r-M+1}\leq r$ strongly depends on the election of the pair of flags $\cF, \cF'\in\cC$ giving the minimum distance of the code. On the other hand, if $d_f(\cC)=d_f(\bar{\cF}, \bar{\cF'})$, for another pair of flags $\bar{\cF}, \bar{\cF'}\in\cC$, following the proof of Proposition \ref{prop: lower bound distance M-disjoint}, one might obtain another lower bound for $d_f(\cC)$ as the sum of the (positive) distances of $r-(M-1)$ different projected codes of $\cC$. 

\begin{corollary}
Let $\cC\subseteq \cF_q(\type, n)$ be an $M$-disjoint flag code for some $1\leq M\leq r$. Then it holds
$$
d_f(\cC)\geq \min \left\lbrace \sum_{j=1}^{r-(M-1)} d_S(\cC_{i_j}) \ \big| \ 1\leq i_1 < \dots < i_{r-(M-1)}\leq r  \ \text{with} \  d_S(\cC_{i_j})\neq 0 \right\rbrace.
$$
In particular, we have that $d_f(\cC)\geq 2(r-(M-1))$.
\end{corollary}

Observe that, if $\cC\subset \cF_q(\type, n)$ is a disjoint flag code, i.e, $1$-disjoint in our new notation, the previous bound coincides with the one given in (\ref{eq: bound distance disjoint}). On the other hand, by using the notation introduced in Section \ref{sec: remarkable values of the distance}, we provide the following sufficient condition on the distance of a flag code to ensure some type of disjointness. More precisely, we can conclude that a given flag code is $(i_1, \dots, i_M)$-disjoint just by checking if its minimum distance is greater than the value $D(i_1, \dots, i_M)^{(\type, n)}$. Recall that, as said in Remark \ref{rem: values Ds only depend on the variety}, fixed the flag variety $\cF_q(\type, n)$, these values  only depend on the choice of the indices $1\leq i_1 < \dots < i_M\leq r$. Hence they are independent from any specific flag code and can be computed and stored as parameters associated to $\cF_q(\type, n)$. We use these remarkable distances as follows.

\begin{theorem}\label{theo: d implies disjoint}
Let $\cC\subseteq\cF_q(\type, n)$ be a flag code such that  $d_f(\cC) > D(i_1, \dots, i_M)^{(\type, n)},$ for some choice of $1\leq M\leq r$ indices $1\leq i_1 < \dots < i_M\leq r$. Then $\cC$ is $(i_1, \dots, i_M)$-disjoint.
\end{theorem}
\begin{proof}
Assume that $\cC$ is not $(i_1, \dots, i_M)$-disjoint for this particular choice of indices $1\leq i_1 < \dots < i_M\leq r$. Then we can find different flags $\cF, \cF'\in\cC$ such that $\cF_{i_j}=\cF'_{i_j}$ for every $1\leq j\leq M$. As a result, the distance vector associated to the pair of flags $\cF$ and $\cF'$ has, at least, $M$ zeros in the positions $i_1, \dots, i_M$. As a result, and according to the definition of $D(i_1, \dots, i_M)^{(\type, n)}$, we have
$$
d_f(\cC) \leq d_f(\cF, \cF') \leq D(i_1, \dots, i_M)^{(\type, n)}, 
$$
which is a contradiction.
\end{proof}

The previous result leads to a sufficient condition for flag codes to be $M$-disjoint in terms of their minimum distance.

\begin{corollary}\label{cor: dist d implies M-disjoint}
Let $\cC\subseteq\cF_q(\type, n)$ be a flag code and consider an integer $1\leq M\leq r.$ If 
$$
d_f(\cC) > \max \left\lbrace D(i_1, \dots, i_M)^{(\type, n)} \ \big| \ 1\leq i_1 < \dots < i_M\leq r \right\rbrace,
$$
then $\cC$ is $M$-disjoint.
\end{corollary}

\begin{remark}
Observe that comparing the distance of a code with the maximum of the values $D(i_1, \dots, i_M)^{(\type, n)}$ is not a big deal since, as said in Remark \ref{rem: values Ds only depend on the variety}, for each choice of indices and type vector, this maximum value can be computed in advance. Moreover, in case of working with full flags on $\bbF_q^n$, this maximum value is explicitly computed in Proposition \ref{prop: D(1,...M) is the largest}. Hence, we can give an easier condition to guarantee that a given full flag code is $M$-disjoint as follows.
\end{remark}

\begin{corollary}\label{cor: dist d implies M-disjoint full type}
Let $\cC$ be a full flag code on $\bbF_q^n$. If $d_f(\cC) > D^{(n-M)}$ for some $1\leq M\leq n-1$, then  $\cC$ is $M$-disjoint.
\end{corollary}

Theorem \ref{theo: d implies disjoint} and Corollary \ref{cor: dist d implies M-disjoint} state sufficient conditions to deduce some degree of disjointness in terms of the minimum distance of a flag code. The concept of disjointness and, in particular, these two results will be crucial to establish bounds for the cardinality of flag codes in $\cF_q(\type, n)$ with a prescribed minimum distance.

\section{Bounds for the cardinality of flag codes}\label{sec: bounds}

This section is devoted to give upper bounds for the cardinality of flag codes from arguments introduced in both Sections \ref{sec: remarkable values of the distance} and \ref{sec: disjointness}. As said in Section \ref{sec: prelim}, the value $A_q^f(n, d, \type)$ denotes the maximum possible size for flag codes in $\cF_q(\type, n)$ with distance $d$. In the particular case of full flags on $\bbF_q^n$, we just write $A_q^f(n, d)$.  Up to the moment, bounds for $A_q^f(n, d)$ have only been studied in \cite{Kurz20}. In that paper, the author develops techniques to determine upper bounds for the size of full flag codes and gives an exhaustive list of them for small values of $n$. Out of the full type case, the author also exhibits some concrete examples. The bounds in the present paper are valid for any type vector and arise from  different techniques. More precisely, for each value of the distance, we apply Theorem \ref{theo: d implies disjoint} and Corollaries \ref{cor: dist d implies M-disjoint} and \ref{cor: dist d implies M-disjoint full type}, in order to ensure certain degree of disjointness and derive upper bounds for $A_q^f(n, d, \type)$, related to the size of a suitable flag variety.

From now on, we will write $d$ to denote a possible distance between flags in $\cF_q(\type, n)$, that is, an even integer $d\in [0, D^{(\type, n)}]$. Next we will use the values $D(i_1, \dots, i_M)^{(\type, n)}$ defined in Section \ref{sec: remarkable values of the distance}, along with the condition of $(i_1, \dots,i_M)$-disjointness introduced in Section \ref{sec: disjointness}, to derive upper bounds for $A_q^f(n, d, \type)$.

\begin{theorem}\label{theo: bound general type}
If $d > D(i_1, \dots, i_M)^{(\type, n)}$, for a particular choice of $1\leq M\leq r$ indices $1\leq i_1 < \dots < i_M\leq r$, then 
$$
A_q^f(n, d, \type)\leq |\cF_q((t_{i_1}, \dots, t_{i_M}), n)| = \begin{bmatrix} n \\ t_1 \end{bmatrix}_q \begin{bmatrix} n -t_1 \\ t_2-t_1 \end{bmatrix}_q \cdots \begin{bmatrix} n-t_{r-1} \\ n-t_r \end{bmatrix}_q.
$$
\end{theorem}
\begin{proof}
By application of Theorem \ref{theo: d implies disjoint}, we know that every flag code $\cC\subseteq\cF_q(\type, n)$ with distance $d>D(i_1, \dots, i_M)^{(\type, n)}$ must be $(i_1, \dots, i_M)$-disjoint. Hence, it holds 
$$
|\cC| = |p_{(i_1, \dots, i_M)}(\cC)| \leq |\cF_q((t_{i_1}, \dots, t_{i_M}), n)|.
$$
 Consequently, every flag code in $\cF_q(\type, n)$ with minimum distance $d >D(i_1, \dots, i_M)^{(\type, n)}$ cannot contain more flags than the flag variety  $\cF_q((t_{i_1}, \dots, t_{i_M}), n).$ The last equality follows from (\ref{eq: card flags}).
\end{proof}

Comparing the distance $d$ with all the possible values of $D(i_1, \dots, i_M)^{(\type, n)}$ leads to the next result, which is a direct consequence of Theorem \ref{theo: bound general type}.
\begin{corollary}\label{cor: bound d > all D(i) type n}
If $d > D(i_1, \dots, i_M)^{(\type, n)}$ for every election of $1\leq M\leq r$ indices $1\leq i_1 < \dots < i_M\leq r$, then 
$$
A_q^f(n, d, \type)\leq \min\{|\cF_q((t_{i_1}, \dots, t_{i_M}), n)| \ | \ 1\leq i_1 < \dots < i_M\leq r \}.
$$
\end{corollary}

Notice that, in the case that $M=1$, Theorem  \ref{theo: bound general type} entails a bound for $A_q^f(n, d, \type)$ in terms of the size of certain Grassmann varieties. 

\begin{corollary}\label{cor: bound d > D(i) type n}
Assume that $d > D(i)^{(\type,n)}$ for some $1\leq i\leq r$. It holds
$$
A_q^f(n, d, \type) \leq |\cG_q(t_i, n)|. 
$$
\end{corollary}

Here below, we provide a potentially tighter bound than the one in Corollary \ref{cor: bound d > D(i) type n} in terms of  the maximum possible size for constant dimension codes in $\cG_q(t_i,n)$ with a suitable value of the subspace distance. Notice that, if $d>D(i)^{(\type,n)}$, by the definition of $D(i)^{(\type,n)}$, no distance vector in $\cD(d,\type, n)$ can have a zero as its $i$-th component. Therefore, the value $\bar{d}_i$ defined in (\ref{def: bar d_i and bar D_i}) satisfies $\bar{d}_i\geq 2$. As a consequence, it makes sense to consider the next upper bound.

\begin{theorem}\label{theo: tighter bound d > D(i) type n}
If $d>D(i)^{(\type,n)}$ for some $1\leq i\leq r$, then
$$
A_q^f(n, d, \type) \leq  A_q(n, \bar{d}_i, i).
$$
\end{theorem}
\begin{proof}
Let $\cC$ be a flag code in $\cF_q(\type, n)$ such that $d=d_f(\cC) > D(i)^{(\type,n)}$ and assume that  $|\cC| >  A_q(n, \bar{d}_i, i)$. By means of Theorem \ref{theo: d implies disjoint}, we know that $\cC$ is $(i)$-disjoint, i.e., $|\cC|=|\cC_i|$. Hence $\cC_i$ is a code in $\cG_q(t_i, n)$ with more than $A_q(n, \bar{d}_i, i)$ subspaces. As a result, we have that $d_S(\cC_i) < \bar{d}_i$. Consequently, there must exist different flags $\cF, \cF'\in\cC$ such that $d_S(\cF_i, \cF'_i)=d_S(\cC_i) < \bar{d}_i$. Proposition \ref{prop: relation with bar d} leads to
$$
d= d_f(\cC) \leq d_f(\cF, \cF') \leq  D(i, d_S(\cF_i, \cF'_i)) < d,
$$
which is a contradiction.
\end{proof}

\begin{remark}\label{rem: if d>2 better bound}
Notice that, since $\bar{d}_i\geq 2$, we clearly have 
$$
A_q(n, \bar{d}_i, i) \leq  A_q(n, 2, i) = |\cG_q(i, n)|
$$
and the equality holds if, and only if, $\bar{d}_i=2$. Consequently, the upper bound for $A_q^f(n, d, \type)$ given in Theorem \ref{theo: tighter bound d > D(i) type n} is as good as the one provided in Corollary \ref{cor: bound d > D(i) type n} and it is even tighter in case $\bar{d}_i\geq 4.$ 
\end{remark}

Let us consider now the full flag variety. To do so, from now on, we will write $d$ to denote a feasible distance between full flags on $\bbF_q^n$, i.e., an even integer with $0\leq d\leq D^n$. In this case, all the results in this section still hold true. However, since we have a better description of the values $D(i_1,\dots, i_M)^n$ when we consider the full flag variety, we can give more information for this specific case. For instance, fixed $1\leq M\leq n-1$, instead of checking the condition $d>D(i_1, \dots, i_M)^n$ for every choice of indices as in Corollary \ref{cor: bound d > all D(i) type n}, by means of Proposition \ref{prop: D(1,...M) is the largest},  one just need to ascertain if $d> D^{n-M}$ holds. Moreover, when restricting to the case $M=1$, by means of Proposition \ref{prop: D(i)=D(n-i)}, we can restrict ourselves to indices $1\leq i\leq \lfloor\frac{n}{2}\rfloor$.

The next result follows straightforwardly from the definition of the value $D(i)^n$ (see (\ref{eq: summands D(i)})) along with Propositions \ref{prop: D(i)=D(n-i)} and \ref{prop: Order on the D(i)^n}. 
\begin{lemma}
If $d > D(i)^n$ for some $1\leq i\leq \lfloor\frac{n}{2}\rfloor$, then the values $\bar{d}_j$ defined in (\ref{def: bar d_i and bar D_i}) satisfy
$$
\bar{d}_j \geq 2, \ \text{for every} \ i\leq j\leq n-i.
$$
\end{lemma}

By means of the previous lemma, and arguing as in Theorem \ref{theo: tighter bound d > D(i) type n}, whenever $d> D(i)^n$ holds, we obtain the next upper bound for $A_q^f(n, d)$.

\begin{theorem}\label{theo: tighter bound full}
If $d>D(i)^n$ for a given $1\leq i\leq \lfloor\frac{n}{2}\rfloor$, then
$$
A_q^f(n, d) \leq \min\{ A_q(n, \bar{d}_j, j)  \ | \ i\leq j\leq n-i\}.
$$
\end{theorem}

Using this last result when working with full flags gives us a bound as good as the one given in Theorem \ref{theo: tighter bound d > D(i) type n}, formulated for the general type. Moreover, in some cases, it even improves it, as we can see in the following example.

\begin{example}
For $n=6$ and the full type vector, consider the flag distance $d=16$, which satisfies $d=16 > D(1)^6=12$. Moreover, taking into account that $\cD(16, 6)=\{(2,4,4,4,2)\}$, it is clear that $\bar{d}_i=2$ for $i=1,5$ and $\bar{d_j}=4$ for $j=2,3,4$. Hence, Theorem \ref{theo: tighter bound d > D(i) type n}, leads to
$$
A_q^f(6,16) \leq  A_q(6, 2, 1) = |\cG_q(6,1)| = q^5+q^4+q^3+q^2+q+1
$$
(see (\ref{eq: card grassmannian})). On the other hand, by using Theorem \ref{theo: tighter bound full}, we obtain
$$
A_q^f(6,16) \leq A_q(6,4,2) = q^4+q^2+1,
$$
which improves the previous bound. Notice that the last equality just gives the cardinality of any $2$-spread code in $\bbF_q^6$, i.e., optimal constant dimension codes (of dimension $2$) having the maximum distance. These codes were introduced in \cite{ManGorRos2008}.
\end{example}

\section{A complete example}\label{sec: example}

In this section we illustrate how to combine all the elements introduced in this paper in order to exhibit relevant information about a flag code with a prescribed minimum distance $d$. To do so, we compute all the values $D(i_1, \dots, i_M)^{(\type, n)}$, defined in Section \ref{sec: remarkable values of the distance} for a specific choice of $n$ and $\type$. 

Let us fix $n=7$ and consider both the full type vector $(1, 2, 3, 4, 5, 6)$ and the type vector $\type=(t_1, t_2, t_3, t_4)=(1, 3, 5, 6)$. We start working with full flags and computing all the values $D(i_1, \dots, i_M)^7$, for every possible choice $ 1 \leq M \leq 6$ and indices $1\leq i_1 < \dots < i_M \leq 6$. As pointed out in Section \ref{sec: remarkable values of the distance}, these values are not only useful for the full type case  but also serve to extract conclusions for any  other flag variety on $\bbF_q^7$ (see Proposition \ref{prop: projection dist vectors many zeros}).

The following table shows all these distances $D(i_1, \dots, i_M)^7$, separated according to the number of zeros $1\leq M\leq 6$. We also exhibit the associated distance vector $\bD(i_1, \dots, i_M)^7$, the choice of ordered indices $1\leq i_1 < \dots < i_M\leq 6$ and the multiset of differences $\{\{ i_1, i_2-i_1, \dots, 7-i_M \}\}$. Recall that, as stated in Remark \ref{rem: ordered differences}, we can restrict ourselves to families of indices such that the differences $1\leq i_1 \leq i_2-i_1 \leq \dots \leq 7-i_M$ are also ordered. Any other choice of indices $1\leq i_1 < \dots < i_M\leq 6$ has an associated multiset of differences $\{\{ i_1, i_2-i_1, \dots, 7-i_M\}\}$ that already appears in these tables. For instance, to compute the value $D(1, 3,6)^7$, we just need to consider the  multiset
$$
\{\{1, 3-1, 6-3, 7-6\}\} = \{\{ 1, 2, 3, 1\}\}
$$ 
and order its elements as an increasing sequence $\{\{1, 1, 2, 3\}\}$. This multiset already appears in Table \ref{table: values D(i1, ..., iM)7}, associated to the choice of indices $(1,2,4)$. Hence, 
$$
D(1,3,6)^7=D(1,2,4)^7= 6.
$$

\begin{table}[H]
\begin{center}
\begin{tabular}{cccc}
\hline
\multicolumn{4}{|c|}{\cellcolor{Gray}$\mathbf{M=1}$}                                                                                              \\ \hline
\multicolumn{1}{|c}{$i_1$}           & \multicolumn{1}{|c}{Differences }                     & \multicolumn{1}{|c}{$\bD(i_1)^7$         } & \multicolumn{1}{|c|}{$D(i_1)^7$} \\ \hline
\multicolumn{1}{|c}{$1$}               & \multicolumn{1}{|c}{$\{\{1, 6\}\}$  }                    & \multicolumn{1}{|c}{$(0,2,4,6,4,2)$}       & \multicolumn{1}{|c|}{18}         \\ \hline
\multicolumn{1}{|c}{$2$}               & \multicolumn{1}{|c}{$\{\{2,5\}\}$ }                       & \multicolumn{1}{|c}{$(2,0,2,4,4,2)$ }      & \multicolumn{1}{|c|}{14}         \\ \hline
\multicolumn{1}{|c}{$3$}               & \multicolumn{1}{|c}{$\{\{3,4\}\}$ }                       & \multicolumn{1}{|c}{$(2,2,0,2,4,2)$}       & \multicolumn{1}{|c|}{12}         \\ \hline
\multicolumn{4}{|c|}{\cellcolor{Gray}$\mathbf{M=2}$}                                                                                              \\ \hline
\multicolumn{1}{|c}{$(i_1, i_2)$}           & \multicolumn{1}{|c}{Differences}                      & \multicolumn{1}{|c}{$\bD(i_1, i_2)^7$}        & \multicolumn{1}{|c|}{$D(i_1, i_2)^7$} \\ \hline
\multicolumn{1}{|c}{$(1,2)$}               & \multicolumn{1}{|c}{$\{\{1,1,5\}\}$}                     & \multicolumn{1}{|c}{$(0,0,2,4,4,2)$}       & \multicolumn{1}{|c|}{12}         \\ \hline
\multicolumn{1}{|c}{$(1,3)$}               & \multicolumn{1}{|c}{$\{\{1,2,4\}\}$}                    & \multicolumn{1}{|c}{$(0,2,0,2,4,2)$}       & \multicolumn{1}{|c|}{10}         \\ \hline
\multicolumn{1}{|c}{$(1,4)$}               & \multicolumn{1}{|c}{$\{\{1, 3,3\}\}$}                        & \multicolumn{1}{|c}{$(0,2,2,0,2,2)$ }      & \multicolumn{1}{|c|}{8}         \\ \hline
\multicolumn{1}{|c}{$(2,4)$}     & \multicolumn{1}{|c}{$\{\{2,2,3\}\}$} & \multicolumn{1}{|c}{$(2,0,2,0,2,2)$} & \multicolumn{1}{|c|}{8}           \\ \hline
\multicolumn{4}{|c|}{\cellcolor{Gray}$\mathbf{M=3}$}                                                                                              \\ \hline
\multicolumn{1}{|c}{$(i_1, i_2, i_3)$}           & \multicolumn{1}{|c}{Differences  }                    & \multicolumn{1}{|c}{$\bD(i_1, i_2, i_3)^7$ }         & \multicolumn{1}{|c|}{$D(i_1, i_2, i_3)^7$} \\ \hline
\multicolumn{1}{|c}{$(1,2,3)$}               & \multicolumn{1}{|c}{$\{\{ 1,1,1,4\}\}$}            & \multicolumn{1}{|c}{$(0,0,0,2,4,2)$} & \multicolumn{1}{|c|}{8}           \\ \hline
\multicolumn{1}{|c}{$(1,2,4)$}               & \multicolumn{1}{|c}{$\{\{1,1,2,3\}\}$}            & \multicolumn{1}{|c}{$(0,0,2,0,2,2)$} & \multicolumn{1}{|c|}{6}           \\ \hline
\multicolumn{1}{|c}{$(1,3,5)$}               & \multicolumn{1}{|c}{$\{\{1,2,2,2 \}\}$ }            & \multicolumn{1}{|c}{$(0,2,0,2,0,2)$} & \multicolumn{1}{|c|}{6}           \\ \hline
\multicolumn{4}{|c|}{\cellcolor{Gray}$\mathbf{M=4}$}                                                                                              \\ \hline
\multicolumn{1}{|c}{$(i_1, \dots, i_4)$}           & \multicolumn{1}{|c}{Differences }                     & \multicolumn{1}{|c}{$\bD(i_1, \dots, i_4)^7$ }         & \multicolumn{1}{|c|}{$D(i_1, \dots, i_4)^7$} \\ \hline
\multicolumn{1}{|c}{$(1,2,3,4)$}               & \multicolumn{1}{|c}{$\{\{1,1,1,1,3\}\}$}            & \multicolumn{1}{|c}{$(0,0,0,0,2,2)$} & \multicolumn{1}{|c|}{4}           \\ \hline
\multicolumn{1}{|c}{$(1,2,3,5)$}               & \multicolumn{1}{|c}{$\{\{ 1,1,1,2,2\}\}$}            & \multicolumn{1}{|c}{$(0,0,0,2,0,2)$} & \multicolumn{1}{|c|}{4}           \\ \hline
\multicolumn{4}{|c|}{\cellcolor{Gray}$\mathbf{M=5}$}                                                                                              \\ \hline
\multicolumn{1}{|c}{$(i_1, \dots, i_5)$}           & \multicolumn{1}{|c}{Differences }                     & \multicolumn{1}{|c}{$\bD(i_1, \dots, i_5)^7$  }        & \multicolumn{1}{|c|}{$D(i_1, \dots, i_5)^7$} \\ \hline
\multicolumn{1}{|c}{$(1,2,3,4,5)$}               & \multicolumn{1}{|c}{$\{\{1,1,1,1,1,2\}\}$}            & \multicolumn{1}{|c}{$(0,0,0,0,0,2)$} & \multicolumn{1}{|c|}{2}           \\ \hline
\multicolumn{4}{|c|}{\cellcolor{Gray} $\mathbf{M=6}$}                                                                                              \\ \hline
\multicolumn{1}{|c}{$(i_1, \dots, i_6)$}           & \multicolumn{1}{|c}{Differences  }                    & \multicolumn{1}{|c}{$\bD(i_1, \dots, i_6)^7$}          & \multicolumn{1}{|c|}{$D(i_1, \dots, i_6)^7$} \\ \hline
\multicolumn{1}{|c}{$(1,\dots, 6)$}               & \multicolumn{1}{|c}{$\{\{1,1,1,1,1,1,1\}\}$}            & \multicolumn{1}{|c}{$(0,0,0,0,0,0,0)$} & \multicolumn{1}{|c|}{0}           \\ \hline
\end{tabular}
\caption{Possible values of $D(i_1, \dots, i_M)^7$ for every $1\leq M \leq 6$.}\label{table: values D(i1, ..., iM)7}
\end{center}
\end{table}

The next table contains upper bounds for $A_q^f (7, d)$, for every value of $2\leq d\leq D^7=24$ and every prime power $q$. To compute them, we compare $d$ with specific values $D(i_1, \dots, i_M)^7$ provided in Table \ref{table: values D(i1, ..., iM)7}, for some $1\leq M\leq 6$. Notice that applying Theorem \ref{theo: bound general type} to different elections either of the integer $M$ or of indices $i_1< \dots < i_M$ provides, in general, different bounds. We proceed as in Corollary \ref{cor: bound d > all D(i) type n} and give the tightest bound for each case. Moreover, observe that the restriction to the families of ordered indices in Table \ref{table: values D(i1, ..., iM)7} is not a problem since any choice of $M$ indices $\{i_1, \dots, i_M\}$ and $\{ j_1, \dots, j_M\}$ giving equal multisets 
$$\{\{ i_1, i_2-i_1, \dots, n-i_M\}\}=\{\{ j_1, j_2-j_1, \dots, n-j_M\}\}$$
also provide the same bound 
$$
A_q^f(n, d) \leq  |\cF_q((i_1, \dots, i_M), n)| = |\cF_q((j_1, \dots, j_M), n)|
$$
because the cardinality of the flag variety 

\begin{equation*}
 \begin{gathered}
|\cF_q((i_1, \dots, i_M), n)|  =     \begin{bmatrix} n \\ i_1 \end{bmatrix}_q \begin{bmatrix} n -i_1 \\ i_2-i_1 \end{bmatrix}_q \cdots \begin{bmatrix} n-i_{M-1} \\ n-i_M \end{bmatrix}_q \nonumber \\
             =   \frac{(q^n-1)\dots (q-1)}{\big( (q^{i_1}-1) \dots (q-1)\big)\Big( \prod_{l=1}^{M} \big( (q^{i_l-i_{l-1}}-1) \dots (q-1)\big)\Big)\big( (q^{n-i_M}-1) \dots (q-1)\big)} 
\end{gathered}
\end{equation*}
just depends on the values $i_1, i_2-i_1, \dots, n-i_M.$

\begin{table}[H]
\begin{center}
\begin{tabular}{|c|l|l|}
\hline
{\cellcolor{Gray}$d$} & \multicolumn{1}{|c|}{\cellcolor{Gray} $D(i_1, \dots, i_M)^7$} &  \multicolumn{1}{|c|}{\cellcolor{Gray} Upper bound for $A_q^f (7, d)$}        \\ \hline
$2$          & $D(1,2,3,4,5,6)^7=0$      & $|\cF_q((1,2,3,4,5,6), 7)|= \dfrac{(q^7-1)\cdots(q^2-1)}{(q-1)^6}$                     \\ [10pt] \hline 
& &  \\ [-1em]
$4$          & $D(1,2,3,4,5)^7=2$      & $|\cF_q((1,2,3,4,5), 7)| = \dfrac{(q^7-1)\cdots(q^3-1)}{(q-1)^5}$ \\ [10pt]\hline 
& &  \\ [-1em]
$6$          & $D(1,2,3,4)^7=4$      & $|\cF_q((1,2,3,4), 7)|= \dfrac{(q^7-1)\cdots(q^4-1)}{(q-1)^4} $      \\ [10pt]\hline 
& &  \\ [-1em] 
$8$          & $D(1, 2, 4)^7=6$          & $|\cF_q((1,2,4), 7)| = \dfrac{(q^7-1)(q^6-1)(q^5-1)(q^2+1)}{(q-1)^3} $   \\ [10pt] \hline  & &  \\ [-1em]
$10$         & $D(1, 4)^7=8$             & $|\cF_q((1,4), 7)| = \dfrac{(q^7-1)(q^5-1)(q^3+1)(q^2+1)}{(q-1)^2}$  \\ [10pt] \hline
& &  \\ [-1em] 
$12$         & $D(1, 3)^7=10$            & $|\cF_q((1,3), 7)|= \dfrac{(q^7-1)(q^5-1)(q^4+q^2+1)}{(q-1)^2} $  \\ [10pt] \hline
& &  \\ [-1em]
$14$         & $D(1,2)^7=12$             & $|\cF_q((1,2), 7)|= \dfrac{(q^7-1)(q^6-1)}{(q-1)^2} $        \\  [10pt]\hline
& &  \\ [-1em]
$16-18$      & $D(2)^7= 14 $             & $|\cG_q(2,7)| = \dfrac{(q^7-1)(q^4+q^2+1)}{(q-1)}$ \\ [10pt] \hline
& &  \\ [-1em]
$20-24$      & $D(1)^7=18$               & $|\cG_q(1,7)|=\dfrac{(q^7-1)}{(q-1)}$  \\[10pt]  \hline
\end{tabular}
\caption{Bounds for $A_q^f(7,d)$ obtained by using Theorem \ref{theo: bound general type}.}\label{table: bounds n=7 full}
\end{center}
\end{table}

Notice that the bounds for $A_q^f(7, d)$ in Table \ref{table: bounds n=7 full} do not change for distances $16\leq d\leq 18$ or $20\leq d \leq 24$. In Table \ref{table: bounds n=7 full 2}, for each flag distance value $16\leq d \leq 24$, we indicate the specific choice  of $1\leq i\leq 6$ and  the corresponding value $\bar{d}_i$ (see (\ref{def: bar d_i and bar D_i})) that provide the best upper bound for $A^f_q(7, d)$ that can be obtained by means of Theorem \ref{theo: tighter bound full}.

\begin{table}[H]
\begin{center}
\begin{tabular}{|c|c|c|l|}
\hline
{\cellcolor{Gray}$d$} & {\cellcolor{Gray} $i$} & {\cellcolor{Gray} $\bar{d}_i$} & \multicolumn{1}{|c|}{\cellcolor{Gray}  Upper bound for $A_q^f (7, d)$}          \\ \hline
& & &  \\ [-1em]
$16$          & 2  & 2           & $A_q(7,2,2)= |\cG_q(2,7)| = \dfrac{(q^7-1)(q^4+q^2+1)}{(q-1)} $  \\ [10pt] \hline
& & &  \\ [-1em]
$18$          &  2 & 2           & $A_q(7,2,2)= |\cG_q(2,7)| = \dfrac{(q^7-1)(q^4+q^2+1)}{(q-1)}$                   \\ [10pt] \hline
& & &  \\ [-1em]
$20$          &  1 & 2           & $A_q(7,2,1)= |\cG_q(1,7)| = \dfrac{(q^7-1)}{(q-1)}$ \\ [10pt] \hline
& & & \\ [-1em]
$22$          &  4 & 4           & $A_q(7,4,2)\leq q(q^4+q^2+1)$        \\ [5pt] \hline
$24$          &  3 & 6           & $A_q(7,6,3)= q^4 +1$   \\ \hline
\end{tabular}
\caption{Bounds for $A_q^f(7,d)$ obtained by using Theorem \ref{theo: tighter bound full}.}\label{table: bounds n=7 full 2}
\end{center}
\end{table}

Notice that, as said in Remark \ref{rem: if d>2 better bound}, for  those cases in which $\bar{d}_i=2$, bounds in Tables \ref{table: bounds n=7 full} and \ref{table: bounds n=7 full 2} coincide. On the other hand, whenever $\bar{d}_i>2$, Table \ref{table: bounds n=7 full 2} gives better bounds. The next example illustrates how bounds in this table have been computed.
\begin{example}
 For $d=20$ and the full flag variety on $\bbF_q^7$, we have
$$
\cD(20,7)=\{(2,4,4,4,4,2), \ (2,2,4,6,4,2), \ (2,4,6,4,2,2)\}.
$$
As a consequence, it holds $\bar{d}_i=2$ for $i=1,2,5,6$ and $\bar{d}_j=4$ for $j=3,4$. Hence, Theorem \ref{theo: tighter bound full} leads to three possible upper bounds for $A_q^f(7, 20)$:
$$
\begin{array}{ccl}
A_q^f(7,20) &\leq & A_q(7,2,1)  = A_q(7,2,6) = |\cG_q(1,7)| = \frac{q^7-1}{q-1},  \\
A_q^f(7,20) &\leq & A_q(7,2,2)  = A_q(7,2,5) = |\cG_q(2,7)| = \frac{(q^7-1)(q^4+q^2+1)}{(q-1)}, \\ 
A_q^f(7,20) &\leq & A_q(7,4,3)  = A_q(7,4,4).
\end{array}
$$
Clearly the first bound is tighter than the second one. Moreover, by means of \cite[Th. 3.20]{TablesSubspaceCodes}, we know that
$$
A_q(7,4,3) \geq q^8+q^5+q^4-q-1 > q^6+\dots +q+1 = \frac{q^7-1}{q-1} = |\cG_q(1,7)|.
$$
Thus, Theorem \ref{theo: tighter bound full} leads to $A_q^f(7,20)\leq |\cG_q(1,7)|,$ as we see in Table \ref{table: bounds n=7 full 2}.
\end{example}

Using similar arguments we arrive to give  some upper bounds for $A_q^f(n,d)$ that coincide with the already presented in \cite{Kurz20}. See, for instance, Propositions 2.5, 2.6, 2.7, 4.4, 4.5, 4.6, 6.1, 6.2, 6.3 and 6.4 in that paper. 

Now, also for $n=7$ but for type vector $\type=(1,3,5,6)$, we apply the results presented in this paper with the goal of exhibiting upper bounds for the cardinality of flag codes of this specific type vector. We start computing the values $D(i_1,\dots, i_M)^{(\type, 7)}$, for $1\leq M\leq 4$, by applying Proposition \ref{prop: projection dist vectors many zeros} to the already computed values  $D(t_{i_1}, \dots, t_{i_M})^7$ in Table \ref{table: values D(i1, ..., iM)7} and their associated vectors $\bD(t_{i_1}, \dots, t_{i_M})^7$.

\begin{table}[H]
\begin{center}
\begin{tabular}{cccc}
\hline
\multicolumn{4}{|c|}{\cellcolor{Gray}$\mathbf{M=1}$}                                                                                              \\ \hline
\multicolumn{1}{|c}{$i_1$}           & \multicolumn{1}{|c}{$\bD(t_{i_1})^7$ }                     & \multicolumn{1}{|c}{$\bD(t_{i_1})^{(\type,7)}$         } & \multicolumn{1}{|c|}{$D(t_{i_1})^{(\type,7)}$} \\ \hline
\multicolumn{1}{|c}{$1$}               & \multicolumn{1}{|c}{$(0,2,4,6,4,2)$  }                    & \multicolumn{1}{|c}{$(0,4,4,2)$}       & \multicolumn{1}{|c|}{10}         \\ \hline
\multicolumn{1}{|c}{$2$}               & \multicolumn{1}{|c}{$(2,2,0,2,4,2)$ }                       & \multicolumn{1}{|c}{$(2,0,4,2)$ }      & \multicolumn{1}{|c|}{8}         \\ \hline
\multicolumn{1}{|c}{$3$}               & \multicolumn{1}{|c}{$(2,4,4,2,0,2)$ }                       & \multicolumn{1}{|c}{$(2,4,0,2)$}       & \multicolumn{1}{|c|}{8}         \\ \hline
\multicolumn{1}{|c}{$4$}               & \multicolumn{1}{|c}{$(2,4,6,4,2,0)$ }                       & \multicolumn{1}{|c}{$(2,6,2,0)$}       & \multicolumn{1}{|c|}{10}         \\ \hline
\multicolumn{4}{|c|}{\cellcolor{Gray}$\mathbf{M=2}$}                                                                                              \\ \hline
\multicolumn{1}{|c}{$(i_1, i_2)$}           & \multicolumn{1}{|c}{$\bD(t_{i_1}, t_{i_2})^7$ }                     & \multicolumn{1}{|c}{$\bD(i_1, i_2)^{(\type,7)}$} & \multicolumn{1}{|c|}{$D(i_1, i_2)^{(\type,7)}$} \\ \hline
\multicolumn{1}{|c}{$(1,2)$}               & \multicolumn{1}{|c}{$(0,2,0,2,4,2)$}                     & \multicolumn{1}{|c}{$(0,0,4,2)$}       & \multicolumn{1}{|c|}{6}         \\ \hline
\multicolumn{1}{|c}{$(1,3)$}               & \multicolumn{1}{|c}{$(0,2,4,2,0,2)$}                    & \multicolumn{1}{|c}{$(0,4,0,2)$}       & \multicolumn{1}{|c|}{6}         \\ \hline
\multicolumn{1}{|c}{$(1,4)$}               & \multicolumn{1}{|c}{$(0,2,4,4,2,0)$}                        & \multicolumn{1}{|c}{$(0,4,2,0)$ }      & \multicolumn{1}{|c|}{6}         \\ \hline
\multicolumn{1}{|c}{$(2,3)$}     & \multicolumn{1}{|c}{$(2,2,0,2,0,2)$} & \multicolumn{1}{|c}{$(2,0,0,2)$} & \multicolumn{1}{|c|}{4}           \\ \hline
\multicolumn{1}{|c}{$(2,4)$}               & \multicolumn{1}{|c}{$(2,2,0,2,2,0)$}                     & \multicolumn{1}{|c}{$(2,0,2,0)$}       & \multicolumn{1}{|c|}{4}         \\ \hline
\multicolumn{1}{|c}{$(3,4)$}               & \multicolumn{1}{|c}{$(2,4,4,2,0,0)$}                     & \multicolumn{1}{|c}{$(2,4,0,0)$}       & \multicolumn{1}{|c|}{6}         \\ \hline
\multicolumn{4}{|c|}{\cellcolor{Gray}$\mathbf{M=3}$}                                                                                              \\ \hline
\multicolumn{1}{|c}{$(i_1, i_2, i_3)$}           & \multicolumn{1}{|c}{$\bD(t_{i_1}, t_{i_2}, t_{i_3})^7$ }                     & \multicolumn{1}{|c}{$\bD(i_1, i_2, i_3)^{(\type,7)}$} & \multicolumn{1}{|c|}{$D(i_1, i_2, i_3)^{(\type,7)}$} \\ \hline
\multicolumn{1}{|c}{$(1,2,3)$}               & \multicolumn{1}{|c}{$(0,2,0,2,0,2)$}            & \multicolumn{1}{|c}{$(0,0,0,2)$} & \multicolumn{1}{|c|}{2}           \\ \hline
\multicolumn{1}{|c}{$(1,2,4)$}               & \multicolumn{1}{|c}{$(0,2,0,2,2,0)$}            & \multicolumn{1}{|c}{$(0,0,2,0)$} & \multicolumn{1}{|c|}{2}           \\ \hline
\multicolumn{1}{|c}{$(1,3,4)$}               & \multicolumn{1}{|c}{$(0,2,4,2,0,0)$ }            & \multicolumn{1}{|c}{$(0,4,0,0)$} & \multicolumn{1}{|c|}{4}           \\ \hline
\multicolumn{1}{|c}{$(2,3,4)$}               & \multicolumn{1}{|c}{$(2,2,0,2,0,0)$ }            & \multicolumn{1}{|c}{$(2,0,0,0)$} & \multicolumn{1}{|c|}{2}           \\ \hline
\multicolumn{4}{|c|}{\cellcolor{Gray}$\mathbf{M=4}$}                                                                                              \\ \hline
\multicolumn{1}{|c}{$(i_1, i_2, i_3, i_4)$}           & \multicolumn{1}{|c}{$\bD(t_{i_1}, t_{i_2}, t_{i_3}, t_{i_4})^7$ }                     & \multicolumn{1}{|c}{$\bD(i_1, i_2, i_3, i_4)^{(\type,7)}$} & \multicolumn{1}{|c|}{$D(i_1, i_2, i_3, i_4)^{(\type,7)}$} \\ \hline
\multicolumn{1}{|c}{$(1,2,3, 4)$}               & \multicolumn{1}{|c}{$(0,2,0,2,0,0)$}            & \multicolumn{1}{|c}{$(0,0,0,0)$} & \multicolumn{1}{|c|}{0}           \\ \hline
\end{tabular}
\caption{Possible values of $D(i_1, \dots, i_M)^{(\type, 7)}$.}
\end{center}
\end{table}

Using this table and applying Corollary \ref{cor: bound d > D(i) type n}, we obtain the next list of bounds for $A_q^f(7, d, \type)$. As before, we provide the tightest possible upper bound for each value $d$. We do so by making a suitable choice of $1\leq M\leq 4$ and indices $1\leq i_1<\dots <i_M\leq 4$. This information  is collected in the next table.

\begin{table}[H]
\begin{center}
\begin{tabular}{|c|l|l|}
\hline
{\cellcolor{Gray}$d$} & \multicolumn{1}{|c|}{\cellcolor{Gray} $D(i_1, \dots, i_M)^{(\type, 7)}$} & \multicolumn{1}{|c|}{\cellcolor{Gray} Upper bound for $A_q^f (7, d, \type)$}          \\ \hline
& &  \\ [-1em]
$2$          & $D(1,2,3,4)^{(\type, 7)}=0$ & $|\cF_q(\type, 7)|= \dfrac{(q^7-1)(q^6-1)(q^5-1)(q^3-1)(q^2+1)}{(q-1)^4}$   \\ [10pt] \hline
& &  \\ [-1em]
$4$          & $D(1,2,4)^{(\type, 7)}=2$  & $|\cF_q((1,3,6), 7)| = \dfrac{(q^7-1)(q^6-1)(q^5-1)(q^2+1)}{(q-1)^3}$    \\ [10pt] \hline
& &  \\ [-1em]
$6$          & $D(2,4)^{(\type, 7)}=4$ & $|\cF_q((3,6), 7)|=  \dfrac{(q^7-1)(q^5-1)(q^3+1)(q^2+1)}{(q-1)^2} $        \\ [10pt] \hline
& &  \\ [-1em]
$8$          & $D(3,4)^{(\type, 7)}=6$ & $|\cF_q((5,6), 7)| = \dfrac{(q^7-1)(q^6-1)}{(q-1)^2} $   \\ [10pt] \hline
& &  \\ [-1em]
$10$         & $D(3)^{(\type, 7)}=8$ & $|\cG_q(5, 7)| = \dfrac{(q^7-1)(q^4+q^2+1)}{(q-1)}$   \\ [10pt] \hline
& &  \\ [-1em]
$12-14$         & $D(1)^{(\type, 7)}=10$ & $|\cG_q(1, 7)|= \dfrac{q^7-1}{q-1}$  \\ [10pt] \hline
\end{tabular}
\caption{Bounds for $A_q^f(7,d, \type)$ obtained by using Theorem \ref{theo: bound general type}. }\label{table: bounds n=7}
\end{center}
\end{table}

Last, for distance $d=14=D^{(\type, 7)}$, we can improve the previous bound. Observe that
$$
\cD(14, \type, 7)= \{ \bD^{(\type, 7)} \} = \{(2,6,4,2)\}.
$$
Thus, taking into account that $\bar{d}_2=6$, by using Theorem \ref{theo: tighter bound d > D(i) type n}, we obtain
$$
A_q^f(7,14, \type) \leq A_q(7,6,t_2) = A_q(7,6,3)= q^4+1,
$$
(see \cite[Th. 3.43]{TablesSubspaceCodes} for the last equality) which is a better bound than the one given in Table \ref{table: bounds n=7}.

\section{Conclusions}
In this paper we have addressed an exhaustive study of the flag distance parameter. To do so, we have introduced the concept of distance vector as a tool to represent how a flag distance value can be obtained from different combinations of subspace distances. Besides, we have characterized distance vectors in terms of certain conditions satisfied by their components.

We have presented the class of $(i_1, \dots, i_M)$-disjoint flag codes, as a generalization of the notion of disjointness given in \cite{CasoPlanar} and also established a connection between the property of being $(i_1, \dots, i_M)$-disjoint and the impossibility of having distance vectors with $M$ zeros, placed in the positions $i_1, \dots, i_M.$ This allows us to read some structural properties of flag codes in terms of their minimum distance and their sets of distance vectors. As a consequence of our study, we deduce upper bounds for the value $A^f_q(n, d, \type)$ for every choice of the parameters. These bounds strongly depend on the number of subspaces that can be shared by different flags of a code in $\cF_q(\type, n)$ with minimum distance $d$. We finish our work by explicitly computing our bounds for $A^f_q(7,d, \type)$ and two particular type vectors when we sweep all the possible distance values in each case.

\end{document}